\documentclass[11pt, a4paper]{amsart}
\usepackage{amsmath, amssymb, amsthm, amsfonts, mathtools, url, tikz}
\usepackage[titletoc, toc,page]{appendix}
\usepackage{slashed}
\usepackage[a4paper]{geometry}
\usepackage[colorlinks=true]{hyperref}
\usepackage[english]{babel}
\usepackage[utf8]{inputenc}
\usepackage{chngcntr}
\usepackage{nameref}
\usepackage{yfonts}
\numberwithin{equation}{section}
\newtheorem{theorem}{Theorem}
\newtheorem{definition}[theorem]{Definition}
\newtheorem{lemma}[theorem]{Lemma}
\newtheorem{proposition}[theorem]{Proposition}
\newtheorem{remark}[theorem]{Remark}
\newtheorem{corollary}[theorem]{Corollary}

\numberwithin{theorem}{section}

\newcommand{\R}{\mathbb{R}}
\newcommand{\p}{\partial}
\newcommand{\A}{\mathcal{A}}
\newcommand{\mn}{{\mu \nu}}
\newcommand{\tee}{{\mathbb{T}^d}}
\newcommand{\lbar}{\underline{L}}

\newcommand{\pgood}{\overline{\partial}}
\newcommand{\obar}{{\underline{4}}}
\newcommand{\abar}{{\underline{a}}}
\newcommand{\bbar}{{\underline{b}}}
\newcommand{\cbar}{{\underline{c}}}
\newcommand{\dbar}{{\underline{d}}}
\newcommand{\abbar}{{\underline{a} \, \underline{b}}}
\newcommand{\gzero}{\overline{G}}
\newcommand{\ubar}[1]{{\underline{#1}}}
\newcommand{\tr}{\text{tr}}
\newcommand{\tbox}{\widetilde{\Box}}
\newcommand{\UU}{{\mathcal{U} \mathcal{U}}}
\newcommand{\LL}{{\mathcal{L} \mathcal{L}}}
\newcommand{\TU}{{\mathcal{T} \mathcal{U}}}
\newcommand{\LT}{{\mathcal{L} \mathcal{T}}}

\newcommand{\U}{\mathcal{U}}
\newcommand{\K}{\mathcal{K}}
\newcommand{\vK}{\vert K \vert}
\newcommand{\qv}{\vert q \vert}
\newcommand{\I}{\vert I \vert}
\newcommand{\J}{\vert J \vert}

\allowdisplaybreaks[1]

\title{The Weak Null Condition and Kaluza-Klein Spacetimes}
\author{Zoe Wyatt}
\address{School of Mathematics\\
The University of Edinburgh\\
James Clerk Maxwell Building\\
Peter Guthrie Tait Road\\
Edinburgh\\
EH9 3FD\\
United Kingdom}
\email{zoe.wyatt@ed.ac.uk}
\date{\today}
\begin{document}

\begin{abstract}
In this paper we prove the non-linear stability of a system of non-linear wave equations satisfying the weak null condition. In particular, this includes the case of the non-linear stability of Minkowski spacetime times a $d$-torus subject to perturbations depending only on the non-compact coordinates. Our argument very closely follows the proof of the non-linear stability of Minkowski spacetime by Lindblad and Rodnianski \cite{LR:04}. 
\end{abstract}
\maketitle

\section{Introduction}
Theories of higher-dimensional gravity are of great interest to string theorists. Some of these additional dimensions are typically compactified in order for these theories to reflect the perceived $(3+1)-$dimensional universe. One such theory is that of Kaluza-Klein, where, in the simplest case, $(3+1+1)-$dimensional gravity is compactified on a circle to obtain at low energies a $(3+1)-$dimensional coupled Einstein-Maxwell-Scalar system, see \eqref{eq:intro-motiv-ems}. In an influential work by Witten, it was shown that Kaluza-Klein theory is unstable at the semiclassical level and heuristic arguments were given for classical stability \cite{Witten:1981gj}. In this paper we rigorously prove the classical stability against perturbations depending only on the non-compact coordinates. The precise details are given in Theorem \ref{intro:corol-motiv}, but the main idea is the following:
\begin{theorem} \label{intro:motiv-theorem}
The Minkowski vacuum of the Einstein-Maxwell-Scalar system arising from the zero modes of $(3+d+1)-$dimensional pure Einstein theory compactified on a $\tee$ is non-linearly stable. Furthermore, the radii of the $\tee$ are non-linearly stable to peturbations of the zero modes. 
\end{theorem}
To compare Theorem \ref{intro:motiv-theorem} with the instability in \cite{Witten:1981gj}, recall that  Witten's classical solution, which indicated the semiclassical instability of Kaluza-Klein in the case $d=1$, had initial data with topology $\R^2 \times S^2$. This initial data cannot be imposed on the Kaluza-Klein spacetime, but rather is conjectured to appear spontaneously from a quantum-mechanical process. There is no smooth perturbation of any size, let alone of the small size we consider, that could lead to a change in the $\R^3 \times S^1$ topology of our initial data. Thus the two results do not disagree. 

The non-linear stability we consider is subject to the symmetry assumption that the perturbations only depend on the non-compact directions. This symmetry assumption, also called the zero mode truncation, is, in the physics literature, called consistent since it yields solutions of the full equations of motion of the higher dimensional theory. Indeed initial data obeying this symmetry will yield a solution similarly invariant in the compact directions.  

As frequently done from an effective theory point of view, one can further make a heuristic physical argument that for sufficiently small initial compact radii, it is in fact sufficient to only consider zero mode perturbations \cite{Pope}. Our result shows that the radii are non-linearly stable to zero mode perturbations. Of course from the non-linear PDE point of view, the dynamics from the non-zero modes are still relevant and remain to be understood. Nonetheless, stability of the zero modes is a necessary first step. We now introduce the PDE, see Section \ref{section:intro-motiv-eg} for further discussion on this application to Kaluza-Klein. \\

The vacuum Einstein equations
$$ R_\mn [g] = 0 \,,$$
determine the evolution of a Lorentzian spacetime $(\mathcal{M}, g)$. The initial data for these equations consists of Cauchy 3-surface $\Sigma_0$ with metric  $\gamma_{ij}$ and a symmetric 2-tensor $K_{ij}$, such that the \textit{constraint equations} hold on $\Sigma_0$
\begin{equation} R[\gamma] - K^{ij} K_{ij} + K^i{}_i K^j{}_j = 0 \,, \qquad D_j K^j {}_i - D_i K^j{}_j = 0 \,, \quad i, j \in \{1,2,3 \} \,. \label{eq:intro-vacuum-constrainteq} \end{equation}
Here $D$ is the Levi-Civita connection of $\gamma$. 
One further requires an embedding $\Sigma_0 \subset \mathcal{M}$ such that the pull-back of the solution $g$ to $\Sigma_0$ is $\gamma$ and $K$ is the second fundamental form of $\Sigma_0$. 
It was shown by Choquet-Bruhat and Geroch that for any smooth initial data satisfying the constraint equations there exists a maximal, unique up to diffeomorphism, globally-hyperbolic spacetime $(\mathcal{M},g)$ that is the Cauchy development of the initial data and that satisfies the vacuum Einstein equations \cite{choquet-bruhat1969}.

A natural question to then ask is whether the simplest solution of the vacuum Einstein equations, the Minkowski spacetime, is stable to small perturbations of the initial data $(\R^3, \delta_{ij}, K_{ij}=0)$. This question of the non-linear stability of Minkowski spacetime and the geodesic completeness of the perturbed solution $g_\mn(t)$ was first shown in the monumental work by Christodoulou and Klainerman \cite{CK:93}.

However their proof, and gauge choice, differed significantly from the first proof of local well-posedness by Choquet-Bruhat \cite{bruhat1952}. In this paper, Choquet-Bruhat fixed the coordinate invariance by choosing \textit{wave coordinates} $ \{ x^\mu \}$ which are defined to satisfy the covariant wave equation
$$ g^{\rho \sigma} \nabla_\rho \nabla_\sigma x^\mu = 0 \,,$$
where $\nabla$ is the Levi-Civita connection of $g$. Relative to these wave coordinates the metric satisfies the \textit{wave-coordinate condition}
\begin{equation}
\p_\rho \left( g^{\rho \mu} \sqrt{\vert \det g \vert} \right) = 0 \,. \label{eq:intro-wave-gauge}
\end{equation}
This condition is also known as the harmonic, de Donder, or  wave gauge. In this gauge, the Einstein equations reduce to a system of non-linear wave equations
\begin{equation}
\tbox_g g_\mn := g^{\rho \sigma} \p_\rho \p_\sigma g_\mn = F_\mn (h) (\p h, \p h) \,, \label{eq:intro-ree}
\end{equation}
where $F_\mn (h) (\p h, \p h)$ is an inhomogeneity we will discuss shortly.  Here $\tbox_g = g^{\rho \sigma} \p_\rho \p_\sigma$ is the \textit{reduced} wave operator. Although  the local well-posedness of the system \eqref{eq:intro-ree} was shown in \cite{bruhat1952}, for a long time it was believed that a global-in-time solution could not be achieved \cite{MR0315300}. 

Nonetheless as shown in the seminal work by Lindblad and Rodnianski, it is in fact possible to use the wave gauge to prove the non-linear stability of Minkowski by treating the system \eqref{eq:intro-ree} using an energy argument at the level of the metric tensor \cite{Lindblad:2003hw, LR:04}. 
Our proof very closely follows \cite{LR:04} and we refer to it for a more detailed discussion on the literature and motivation behind their proof.  \\
 
We now briefly summarise the result of \cite{LR:04}. The Einstein equations \eqref{eq:intro-ree} in wave coordinates coupled to a scalar field, written in terms of the perturbation away from Minkowski $h_\mn := g_\mn - m_\mn$, take the form
\begin{subequations}\label{eq:intro-LR-EE}
\begin{equation}
\begin{split}
\tbox_g h_\mn &= F_\mn (h) (\p h, \p h)+ 2 \p_\mu \psi \p _\nu \psi \,, \\
\tbox_g \psi &= 0 \,, \end{split} \label{eq:intro-LR-EE-1}
\end{equation} 
where the inhomogeneity is
\begin{equation} \begin{split}
F_\mn (h) (\p h, \p h) &:= P(\p_\mu h, \p_\nu h) + Q_\mn (\p h, \p h) + G_\mn(h) (\p h, \p h)  \,, \\
P(\p_\mu h, \p_\nu h) &:= \frac{1}{4} \p_\mu h^\rho{}_\rho \p_\nu h^\sigma{}_\sigma - \frac{1}{2} \p_\mu h^{\rho \sigma} \p_\nu h_{\rho \sigma} \,.
\end{split} \end{equation}
\end{subequations}
Note the term involving $\psi$ in \eqref{eq:intro-LR-EE-1} comes from adding a stress energy tensor $T^\psi_\mn := \p_\mu \psi \p_\nu \psi$ to the Einstein equations in the form $R_\mn[g]-\frac{1}{2} g_\mn R[g] = T_\mn$. 
Here $Q_\mn$ is a linear combination of the standard quadratic null forms
\begin{equation} Q_0(\p \phi, \p \psi) := m^\mn \p_\mu \phi \p_\nu \phi \,, \quad  \widetilde{Q}_\mn (\p \phi, \p \psi) := \p_\mu \phi \p_\nu \psi - \p_\nu \phi \p_\mu \psi \,, \label{eq:intro-nullforms} 
\end{equation}
whose behaviour has been studied for some time, see \cite{MR837683}, \cite{nullklainerman} and \cite{nullchristodoulou}. Most importantly both null forms can be estimated by
\begin{equation}
\vert Q ( \p \phi, \p \psi) \vert \leq \vert \pgood \phi \vert \vert \p \psi \vert + \vert \p \phi \vert \vert \pgood \psi \vert \,, \label{eq:intro-null-estimate} 
\end{equation}
where $\p = (\p_t, \nabla)$ are all space-time derivatives and $\pgood$ denotes derivatives tangent to the light cones. 
$G_\mn(h)(\p h, \p h)$ is a term quadratic in $\p h$ with coefficients that smoothly depend on $h$ such that $G(0)(\p h, \p h)=0$. Furthermore define $h^\rho{}_\rho = m^{\rho \mu} h_{\mu \rho}$ and throughout this paper we will raise and lower indices with respect to the background Minkowski metric $m$.

The structure of the non-linearity $F_\mn$ in \eqref{eq:intro-LR-EE}, in particular that of the quadratic non-null form $P_\mn$,  was exploited to great effect in \cite{LR:04}, see also our later Section \ref{section:intro-weak-null} on the weak null condition. 
Furthermore to deal with the ADM mass term at spatial infinity, Lindblad and Rodnianski defined a cut-off function $\chi \in C^\infty$ by
\begin{equation} \chi(s) := \Bigg\lbrace
	\begin{array}{ll} 0\,,  & s \leq 1/2 \\
		1 \,, &  s \geq 3/4 
	\end{array} \,, \label{eq:intro-chi} \end{equation}
and a 2-tensor $\gamma^1_{ij}$ in terms of the ADM mass $M$ by
\begin{equation} \gamma_{ij} = \left( 1+ \chi(r) \frac{M}{r}  \right) \delta_{ij} +  \gamma^1_{ij} \,.
 \label{eq:intro-def-gamma1} \end{equation}
Lindblad and Rodnianski showed that for smooth initial data $( \Sigma_0, \gamma_{ij}, K_{ij}, f, g)$, sufficiently `close' to Minkowski initial data, satisfying the constraint equations,
$$ R[\gamma] - K^{ij} K_{ij} + K^i{}_i K^j{}_j= \vert D f \vert^2 + \vert g \vert^2 \,, \quad D_j K^j {}_i - D_i K^j{}_j = g D_i f \,,$$
the solution $(g_\mn(t), \psi(t))$ to \eqref{eq:intro-LR-EE} can be extended to a global-in-time solution. We will discuss this notion of `closeness' to Minkowski initial data shortly. 
Note here $D$ is the covariant derivative associated to $\gamma$.  \\

In this work we consider the unknowns
\begin{subequations} \label{eq:intro-new-EE}
\begin{equation}
W := \{ h_\mn \}_{\mu, \nu \in \{ 0, 1, 2, 3 \}} \cup \{ \psi_k \}_{k \in \{1, \ldots, m \}} \,, 
\end{equation}
for some $m \in \mathbb{N}$, satisfying the following generalised PDE system:
\begin{equation} \begin{split}
\tbox_g h_\mn &= F_\mn (W) (\p W, \p W) \,, \\
\tbox_g \psi_k &= F_k (W) (\p W, \p W) \,, \\
F_\mn (W) (\p W, \p W) &:= P(\p_\mu W, \p_\nu W) + Q_\mn (\p W, \p W) + G_\mn(W) (\p W, \p W) \,, \\
F_k (W) (\p W, \p W) &:= Q_k (\p W, \p W) + G_k ( W) (\p W, \p W) \,, \\
\end{split}
\end{equation}
together with the wave-coordinate condition
\begin{equation}
\p_\rho \left( g^{\rho \mu} \sqrt{\vert \det g \vert} \right) = 0 \,. \label{eq:intro-new-EE-gauge}
\end{equation}
The quadratic terms $Q_\mn, Q_k$ are linear combinations of the null forms \eqref{eq:intro-nullforms} in terms of $\p W$ variables, contracting the arguments of the null forms with $m_\mn$ and/or arbitrary $N^k \in \R^m$ as appropriate. 
The remaining non-null $\mathcal{O}(( \p W)^2)$ terms are of the form
\begin{equation} \begin{split}
P(\p_\mu W, \p_\nu W) & :=  \frac{1}{4} \p_\mu h^\rho{}_\rho \p_\nu h^\sigma{}_\sigma - \frac{1}{2} \p_\mu h^{\rho \sigma} \p_\nu h_{\rho \sigma}  \\
& \quad + N^k m^{\rho \sigma} \left( \p_{\mu} h_{\rho \sigma}  \p_{\nu} \psi_k +  \p_{\nu} h_{\rho \sigma}  \p_{\mu} \psi_k  \right) + N^{kl} \left( \p_{\mu} \psi_k \p_{\nu} \psi_l + \p_{\nu} \psi_k \p_{\mu} \psi_l \right) \,. 
\end{split} \label{eq:intro-new-EE-P} \end{equation} 
\end{subequations}
Here $N^k \in \R^m, N^{kl} \in \R^{m^2}$ are some arbitrary constant coefficients and $G_\mn, G_k$ are terms quadratic in $\p W$ with coefficients smoothly dependent on $W$ such that $G(0)(\p W, \p W)=0$. Although we have added additional non-linearities to both the $h_\mn$ and $\psi_k$ terms, we have specifically \textit{only} added $\mathcal{O}(( \p W)^2)$ terms which are null forms to $F_k$. This choice is so that the variables $\psi_k$ obey the same estimates as the `best' components of $h_\mn$. See also Remark \ref{remark:loizelet}. 
  
To discuss the initial data for this system further, introduce the Minkowski conformal vector fields
$$ \mathcal{Z} := \left\lbrace \p_\mu \,, \Omega_\mn := x_\mu \p_\nu - x_\nu \p_\mu \,, S:=x^\mu \p_\mu \right\rbrace \,,$$
 and the standard multi-index notation
\begin{align*}
Z^\alpha &:= \sum_{Z \in \mathcal{Z} \,, \sum \vert \alpha_i \vert = \vert \alpha \vert} Z^{\alpha_1} \cdots Z^{\alpha_n} \,,  &\nabla^\alpha & := \sum_{\mu_i \in \{1,2,3\} \,, \sum \vert \alpha_i \vert = \vert \alpha \vert} \p_{\mu_1}^{\alpha_1} \cdots \p_{\mu_n}^{\alpha_n} \,, \\
\vert Z^I \phi \vert &:= \sum_{\vert \alpha \vert \leq \I} \vert Z^\alpha \phi \vert \,, & \vert \nabla^I \phi \vert & := \sum_{\vert \alpha \vert \leq \I} \vert \nabla^\alpha \phi \vert \,.
\end{align*}
The initial data consists of the collection $(\Sigma_0, \gamma_{ij}, K_{ij}, f_k, g_k)$ where $(\Sigma_0, \gamma_{ij})$ is a 3-dimensional Riemannian manifold, $K_{ij}$ a symmetric 2-tensor and $(f_k, g_k)$ are smooth functions. We assume that the initial surface $\Sigma_0$ is diffeomorphic to $\R^3$, has ADM mass $M$, and that there exists a global coordinate chart on $\Sigma_0$ such that as $r := ( m_{ij} x^i x^j )^{1/2} \to \infty$ the initial data satisfies
\begin{equation}
\gamma^1_{ij} \,, f_k = o(r^{-1-\alpha}) \,,  \qquad K_{ij}\,, g_k  = o(r^{-2-\alpha}) \,, \qquad  \alpha >0 \,. \label{eq:intro-decay-id}
\end{equation}
For $N \geq 8$ and fixed constant $\gamma \in (0, 1/2)$, define an initial weighted energy by
\begin{equation} \begin{split} E_N(0) 
& := \sum_{\vert I \vert \leq N}  \int_{\Sigma_0}  w  (1+r)^{2\I} \left( \vert \nabla \nabla^I \gamma^1 \vert^2 + \vert \nabla^I K \vert^2 + \vert \nabla \nabla^I f \vert^2 + \vert \nabla^I g\vert^2 \right) d^3 x \\
& = \sum_{\vert I \vert \leq N}  \int_{\Sigma_0}  (1+r)^{1+2\gamma+2\I} \left( \vert \nabla \nabla^I \gamma^1 \vert^2 + \vert \nabla^I K \vert^2 + \vert \nabla \nabla^I f \vert^2 + \vert \nabla^I g\vert^2 \right) d^3 x \,,
 \end{split} \label{eq:intro-initial-energy} 
\end{equation}
where the weight function $w$ is defined as 
\begin{equation}
w: = w(q) = \Bigg\lbrace \begin{array}{ll}
1+(1+\qv)^{1+2\gamma} \,, & q>0 \\
1+(1+\qv)^{-2 \mu} \,, & q\leq 0
\end{array} \,, \quad q := t-r \,. \label{eq:intro-weight}
\end{equation}
Here $\mu > 0$ is a constant to be fixed later. 
Just as for the initial data \eqref{eq:intro-def-gamma1}, we follow Lindblad and Rodnianski and deal with the ADM mass $M$ at spatial infinity by defining the 2-tensor $h^1_\mn$ 
$$ h_\mn = \chi(r)\chi(r/t) \frac{M}{r} \delta_\mn +  h^1_\mn \,.$$
This is also depicted in Figure \ref{fig:h0} on page \pageref{fig:h0}. The weighted energy for the solution $(h^1_\mn (t), \psi_k(t))$ is defined as
\begin{equation} \begin{split} \mathcal{E}_N[W^1] (t)
:= \sup_{0 \leq \tau \leq t}  \sum_{\vert I \vert \leq N}  \int_{\Sigma_\tau} w \left(  \vert \p Z^I h^1_\mn (\tau, x ) \vert^2 + \vert \p Z^I \psi_k (\tau, x ) \vert^2 \right) d^3 x \,, \label{eq:intro-energy-sol} \end{split}
\end{equation}
where the unknown dynamical variables are
$$W^1 := \{ h^1_\mn \}_{\mu, \nu \in \{ 0, 1, 2, 3 \}} \cup \{ \psi_k \}_{k \in \{1, \ldots, m \}} \,.$$ 
\textbf{Our main result is now:}
\begin{theorem} \label{theorem:intro1}
Let $( \Sigma_0,\gamma_{ij}, K_{ij}, f_k, g_k)$ be smooth initial data for \eqref{eq:intro-new-EE},  asymptotically flat in the sense of \eqref{eq:intro-decay-id}, with $\Sigma_0$ diffeomorphic to $\R^3$. 
There exists a constant $\varepsilon_0>0$ such that for all $\varepsilon \leq \varepsilon_0$ and initial data such that
$$E_N (0)^{1/2} + M \leq \varepsilon \,,$$
the solution $(m_\mn + h_\mn(t)\,, \psi_k(t))$ to the system \eqref{eq:intro-new-EE} can be extended to a global-in-time smooth solution, agreeing with the initial data on $\Sigma_0$
$$ (g_{ij} \vert_{\Sigma_0} = \gamma_{ij} \,, \p_t g_{ij} \vert_{\Sigma_0} = K_{ij} \,, \psi_k \vert_{\Sigma_0} = f_k \,, \p_t \psi_k \vert_{\Sigma_0} = g_k ) \,, $$
and for which the energy satisfies the following inequality for all time
\begin{equation} \mathcal{E}_N[W^1](t) \leq C_N \varepsilon (1+t)^{C\varepsilon} \,. \label{eq:intro-thm-aim} \end{equation}
\end{theorem}

Similar to the method used in \cite{LR:04}, the proof of Theorem \ref{theorem:intro1} relies on a `bootstrap' assumption to make a \textit{continuous induction} argument . 
By standard theory of non-linear wave equations, we can obtain a local-in-time smooth solution $(g_\mn(t)\,, \psi_k (t))$ of our PDE obeying the wave-gauge condition in some maximum interval $0 \leq t \leq T_0 $. This maximum time of existence $T_0< \infty$ is defined by the blow-up of the energy: $\mathcal{E}_N[W^1](t) \to \infty$ as $t \to T_0^-$. 

Furthermore the smallness condition on $E_N(0)^{1/2} +M\leq \varepsilon$ implies that there is some maximal time $T \in (0, T_0)$ on which the following inequality holds
\begin{equation} \mathcal{E}_N[W^1](t) \leq 2 C_N \varepsilon (1+t)^\delta\,,  \quad \forall \, t \in [0,T] \,, \label{eq:intro-bootstrap} \end{equation}
where $\delta \in (0,1/4)$ is some fixed constant with $\delta < \gamma$. The aim of the proof is to contract on the bootstrap assumption \eqref{eq:intro-bootstrap} by choosing $\varepsilon$ sufficiently small. That is, we show that
\begin{equation}\mathcal{E}_N[W^1](t) \leq  C_N \varepsilon^2 (1+t)^{C \varepsilon}\,,  \quad \forall \, t \in [0, T] \,, \,  \varepsilon \leq \varepsilon_0 \,. \label{eq:intro-aim} \end{equation}
By choosing $\varepsilon_0$ sufficiently small, this will imply
$$ \mathcal{E}_N[W^1](t) \leq  C_N \varepsilon^2 (1+t)^{C \varepsilon} < 2 C_N \varepsilon (1+t)^\delta \,, \quad \forall \, t \in [0, T] \,,\,  \varepsilon \leq \varepsilon_0 \,. $$
Since $\mathcal{E}_N[W^1](t)$ is continuous and we have contracted the bootstrap condition, $T$ cannot have been the maximal time for which \eqref{eq:intro-bootstrap} holds, and so we must have $T=T_0$. This however would then imply that $\mathcal{E}_N[W^1](T_0)$ is finite, and so we may extend the solution beyond $T_0$. This contradicts the maximality of $T_0$ and so we have a global-in-time solution with $T_0 = \infty$. Thus the aim of our paper is to show \eqref{eq:intro-thm-aim}.

\subsection{Structure of the Paper}
 First in Section \ref{section:intro-motiv-eg} we discuss the motiving example from Kaluza-Klein theory, the Weak Null condition and some other literature. In Section \ref{section:null-frame} we set up the null frame, and in Section \ref{section:ee-wave-coords} we discuss the generalised PDE system \eqref{eq:intro-new-EE} and its form  when written with respect to this null frame.  Section \ref{section:wave-gauge} is where we derive estimates coming from the wave coordinates and then apply these to the inhomogeneity. In Sections \ref{section:decay-1} and \ref{section:decay-2} we derive the main decay estimates which are then used to derive an integrated energy inequality in Section \ref{section:energy} which concludes our proof of \eqref{eq:intro-aim}. In Appendix \ref{section:ems-full} we derive the non-minimally coupled Einstein-Maxwell-Scalar system, which gives a slightly different approach to the higher dimensional system discussed so far and in Section \ref{section:intro-motiv-eg}. In Appendix \ref{section:hardy} we state some useful identities from \cite{LR:04}. 

\section{Zero Mode Reduction of Kaluza Klein} \label{section:intro-motiv-eg} 
The motivation for considering the system \eqref{eq:intro-new-EE} is to prove Theorem \ref{intro:motiv-theorem}, and we now discuss how this is achieved. For simplicity, we first consider a $(3+1+1)-$dimensional gravitational theory compactified on a circle $S^1$. This was the original set-up considered by Klein \cite{Klein1926} using previous work by Kaluza \cite{Kaluza}. Let the $(3+1+1)-$dimensional metric $G_\mn$ have indices $\mu, \nu \in \{ 0, \ldots, 4 \}$ and satisfy the vacuum Einstein equations
\begin{equation} R_\mn[G] = 0 \,. \label{eq:intro-motiv-ee} \end{equation}
Written with respect to the wave-coordinate condition on the full $(3+1+1)-$dimensional metric, the equation of motion \eqref{eq:intro-motiv-ee} becomes a non-linear wave equation 
\begin{equation} G^{\rho \sigma} \p_\rho \p_\sigma G_\mn = \mathcal{N}(G)(\p G, \p G) \,. \label{eq:intro-motiv-nonlinearwave} \end{equation}
Let $x^a$, for $a=0, \ldots, 3$, denote the coordinates of the non-compact spacetime and $x^\obar$ denote the coordinate of the fifth compact dimension. Since $S^1$ is compact, say of radius $R$, we can Fourier expand the metric
\begin{equation} G_\mn(x^a, x^\obar) = \sum_{n \in \mathbb{Z}} \exp \left( \frac{i n  x^\obar}{R} \right) G_\mn^{(n)}(x^a) \,. \label{eq:intro-motiv-exp} \end{equation}
If we substitute the expansion \eqref{eq:intro-motiv-exp} into the left-hand-side of \eqref{eq:intro-motiv-nonlinearwave} and look at just the terms coming from the flat background $\gzero_\mn$ we obtain
$$ \gzero^{\rho \sigma} \p_\rho \p_\sigma G_\mn = \sum_{n \in \mathbb{Z}} \left( m^{ab} \p_a \p_b - \left( \frac{n}{R} \right)^2 \right) G^{(n)}_\mn \,.$$
Heuristically one can see that the modes $G^{(n)}_\mn$ with $n \neq 0$, will satisfy non-linear Klein-Gordon equations with mass $\vert n \vert /R$. At this point, the standard physical argument is to ignore these $n \neq 0$ modes by taking $R$ sufficiently small that the mass $\vert n \vert /R$ is larger than any probable energy \cite{Pope}. 
From the PDE point of view, including the $n \neq 0$ modes will lead to terms of the form 
$ G^{\obar \, \obar} \p^2 _{\obar \obar}G_\mn$ and thus to the trapping of energy in the compact direction, a much more difficult problem for future work. 

Hence we consider only the $n=0$ modes by setting $G^{(n)}_\mn = 0$ for all $n \neq 0$. This implies the higher-dimensional metric $G_\mn$ depends only on the non-compact coordinates
\begin{equation} G_\mn(x^a, x^\obar) = G_\mn(x^a) \,. \label{eq:intro-motiv-metricassumpt} \end{equation}
Assuming the flat, background metric on $\R^{1+3} \times S^1$, this condition is equivalent to constraining the perturbations $h_\mn$ to satisfy
$$ \p_\obar h_\mn = 0\,,  \quad \forall \, \mu, \nu \in \{ 0, \ldots, 4 \} \,. $$
Indeed under the assumption  \eqref{eq:intro-motiv-metricassumpt}, it is standard, see for example \cite{Pope}, to use an Ansatz involving the $(3+1)-$dimensional metric $g_{ab}$, dilaton $\phi$ and vector potential $\mathcal{A}_a$
\begin{equation} \begin{split} 
G_{ab} & = e^{2 \alpha \phi} g_{ab} + e^{2 \beta \phi} \mathcal{A}_a \mathcal{A}_b \,, \qquad a \,,b \in \{ 0, 1, 2, 3 \} \,,\\
G_{a \obar} & = e^{2 \beta \phi} \mathcal{A}_a\,, \quad G_{\obar \obar} = e^{2 \beta \phi} \,,
\end{split} \label{eq:intro-motiv-ansatz} \end{equation}
where $\alpha := \sqrt{12}/12, \beta := - 2 / \sqrt{12}$ are some  constants. Indeed provided $\beta \neq 0$ this choice fully parametrises the higher-dimensional metric. Using this Ansatz, the higher-dimensional vacuum Einstein equations \eqref{eq:intro-motiv-ee} reduce to the following non-minimally coupled\footnote{We use `non-minimal' here and throughout in the sense that there is non-trivial coupling between the scalar and Maxwell fields.} $(3+1)$-dimensional Einstein-Maxwell-Scalar system:
\begin{subequations} \label{eq:intro-motiv-ems}
\begin{align}
R_{ab} &= \frac{1}{2} \p_a \phi \p_b \phi + \frac{1}{2} e^{-6\alpha \phi} \left( \mathcal{F}_{ac} \mathcal{F}_b {}^c - \frac{1}{4} \mathcal{F}_{cd} \mathcal{F}^{cd} g_{ab} \right) \,, \label{eq:intro-motiv-ems-1} \\
\nabla^a \left( e^{- 6 \alpha \phi} \mathcal{F}_{ab} \right) &= 0 \,, \label{eq:intro-motiv-ems-2} \\
\tbox_g \phi &= - \frac{3}{2} \alpha e^{- 6 \alpha \phi} \mathcal{F}_{cd} \mathcal{F}^{cd}\,, \label{eq:intro-motiv-ems-3}
\end{align} 
\end{subequations} 
where $\mathcal{F}_{ab} := \p_a \mathcal{A}_b - \p_b \mathcal{A}_a$. 
The Ansatz \eqref{eq:intro-motiv-ansatz} is chosen so that $g_{ab}$, $\mathcal{A}_a$ and $\phi$ transform as a $(3+1)-$dimensional metric, vector potential and scalar field respectively. A similar argument can be made when $d >1$, see Appendix \ref{section:ems-full}. 

Indeed when we compactify over a $\tee$, not merely an $S^1$, there will be additional compact coordinates $\{ x^\abar \}$ where $\abar \in \{ 4, \ldots, d+3\}$. While the non-linearities in \eqref{eq:intro-motiv-ems} may exhibit some, perhaps mixed, structure involving null and non-null terms, much more is known about the non-linear structure at the level of the higher-dimensional metric $G_\mn$ and so this is what we turn to now. Impose the flat, background metric $\gzero_\mn$ on $\R^{1+3} \times \tee$ given by
$$ \gzero_\mn := \begin{pmatrix}
m_{ab} & 0 \\ 0 & \delta_{\abar \, \bbar} 
\end{pmatrix} \,, \quad a,b \in \{ 0, 1, 2, 3 \} \,,  \quad \abar, \bbar \in \{ 4, \ldots, d+3\} \,. $$ 
We use Roman letters on the $4-$dimensional non-compact space, underlined Roman letters on the $d$-dimensional compact space and Greek letters $\mu, \nu$ for the full $(3+d+1)-$dimensional space. Note the choice of flat metric implies 
$$ R_\mn [\, \gzero \,] = 0 \,.$$ 
Condition \eqref{eq:intro-motiv-metricassumpt} now becomes
\begin{align}
\p_\abar h_\mn = 0\,, \quad \forall \, \mu \,, \nu \,, \abar \,. \label{eq:TIP}
\end{align}
Let us now turn to the PDE system assuming \eqref{eq:TIP}. If the initial data satisfies \eqref{eq:intro-motiv-metricassumpt}, and the PDE is invariant in the compact directions, then the solution will be also.
The perturbation and inverse perturbation are defined by 
$$ h_\mn := G_\mn - \gzero_\mn \,, H^\mn := G^\mn - \gzero{}^\mn \,, $$
where $G_{\mu \rho } G^{\rho \nu} = \delta^\nu_\mu$.   Differentiating this identity and using \eqref{eq:TIP} implies $\p_\abar H^\mn = 0$.
Imposing the wave coordinate condition \eqref{eq:intro-wave-gauge} on the full metric $G_\mn$, 
\begin{subequations} \label{eq:intro-torus-pde}
\begin{equation}
\p_\rho \left( G^{\rho \mu} \sqrt{\vert \det G \vert} \right) = 0 \,,
\end{equation}
the vacuum Einstein equations reduce to the following non-linear wave system
\begin{equation} \begin{split}
\tbox_g h_{ab} &= P(\p_a h, \p_b h)  + Q_{ab} (\p h, \p h) + G_{ab} (h) (\p h, \p h) \,, \\
\tbox_g h_{a \bbar} &= Q_{a \bbar} (\p h, \p h) + G_{a \bbar} (h) (\p h, \p h) \,, \\
\tbox_g h_\abbar &= Q_\abbar (\p h, \p h) + G_\abbar (h) (\p h, \p h) \,,
\end{split} 
\end{equation}
where the non-null $\mathcal{O}(( \p h )^2)$ terms are
\begin{equation} \begin{split}
P(\p_a h, \p_b h) &  =  m^{ef} m^{cd} \left( \frac{1}{4} \p_a h_{ef} \p_\nu h_{cd} - \frac{1}{2} \p_a h_{ec} \p_b h_{fd} \right) + \delta^\abbar \delta^{\cbar \, \dbar} \left( \frac{1}{4} \p_a h_\abbar \p_b h_{\cbar \, \dbar} - \frac{1}{2}  \p_a h_{\ubar{a} \ubar{c}} \p_b h_{\ubar{b} \ubar{d}} \right) \\
& + \delta^\abbar m^{cd} \left( \frac{1}{4}  \p_a h_\abbar\p_b h_{cd} + \frac{1}{4}  \p_b h_\abbar \p_a h_{cd} - \p_a h_{\abar c} \p_b h_{\bbar d}  \right) \,. \\
\end{split} \label{eq:intro-torus-pde-2}
\end{equation} 
\end{subequations}
As before, $Q_\mn$ consists of linear combinations of the standard null-forms \eqref{eq:intro-nullforms}, and $G_\mn$ is quadratic in $\p h$ and vanishing when $h=0$. The explicit equations are given in Proposition \ref{remark:torus-full-pde}. If we consider the unknowns $\{ h_{a \bbar}, h_\abbar \}$ as variables $\{ \psi_k \}$ then the system \eqref{eq:intro-torus-pde} falls into the class considered in \eqref{eq:intro-new-EE}. Using Theorem \ref{theorem:intro1} we now obtain the following. 

\begin{theorem} \label{intro:corol-motiv}
Let $(\Sigma_0, \gamma_{ij}, K_{ij})$ be smooth initial data for the equations of motion \eqref{eq:intro-torus-pde} arising from the $(3+d+1)-$dimensional vacuum Einstein equations under \eqref{eq:TIP}.  Define $r := ( m_{i'j'} x^{i'} x^{j'} )^{1/2} \to \infty$ and the 2-tensor $\gamma^1_{ij}$ by
\begin{align*}
&\gamma_{ij} = \begin{pmatrix}
\left( 1+ \chi(r) \frac{M}{r} \right) \delta_{i'j'} & 0 \\
0 & \delta_\abbar
\end{pmatrix} + \gamma^1_{ij} \,, \\
&i', j' \in \{1, 2, 3 \} \,, \quad \abar , \bbar \in \{ 4, \ldots, 3+d\} \,, \quad i, j \in \{1, \ldots, 3+d\} \,,
\end{align*} 
with $\chi(r)$ defined in \eqref{eq:intro-chi} and $M \in (0, \infty)$. 
Suppose also that $\Sigma_0$ is diffeomorphic to $\R^3 \times \tee$ and the initial data satisfies the constraint equations \eqref{eq:intro-vacuum-constrainteq} and is asymptotically Kaluza-Klein in the sense that
\begin{equation}
\gamma^1_{ij} = o(r^{-1-\alpha}) \,,  \qquad K_{ij}= o(r^{-2-\alpha}) \,, \qquad  \alpha >0 \,.
\end{equation}
Furthermore for $N \geq 8$ and $\gamma \in (0, 1/2)$ define
$$ E_N(0) := \sum_{\vert I \vert \leq N} \int_{\Sigma_0}  (1+r)^{1+2\gamma+2\I} \left( \vert \nabla \nabla^I \gamma^1 \vert^2  +\vert \nabla^I K \vert^2  \right) d^{3+d} x  \,.$$
Then there exists a constant $\varepsilon_0>0$ such that for all $\varepsilon \leq \varepsilon_0$ and initial data with $E_N(0)^{1/2} + M \leq \varepsilon$, the solution 
$$G_\mn (t) = \gzero_\mn + h_\mn(t) $$
to \eqref{eq:intro-torus-pde} exists for all times and 
in particular $G_\mn (t) \to \gzero_\mn$ as $t \to \infty$. Hence the perturbed $\tee$ radii decay to the radii of the background geometry.  
\end{theorem} 

We assume the initial data required by Theorem \ref{intro:corol-motiv} exists, however the constraint equations have a perhaps more natural interpretation from the perspective of the $(3+1)-$dimensional non-minimally coupled Einstein-Maxwell-Scalar system \eqref{eq:ems-eom}. See Appendix \ref{section:ems-full} for more details. Note also that the radii of the $\tee$ in the background geometry do not necessarily have to be taken small, merely non-zero. 

In fact, in an earlier work by Lindblad and Rodnianski, it was shown that the perturbed solution gives a future geodesically complete manifold. 
Since our variables obey the same decay rates, it should follow similarly to Section 16 in \cite{Lindblad:2003hw} that the perturbed solution $G_\mn(t)$ yields a future causally geodesically complete solution asymptotically converging to $\R^{1+3} \times \tee$.\\

\subsection{Other results and future work}
We now make some final comments on the additional terms considered in the generalised system \eqref{eq:intro-new-EE}. Although our method of proof follows from \cite{LR:04}, the terms $h_\abbar$ and $h_{a \bbar}$ cannot be treated directly as $\psi$ variables in the original system \eqref{eq:intro-LR-EE} of Lindblad and Rodnianski. 
For one, these variables have non-trivial inhomogeneities which is unlike those presented in \eqref{eq:intro-LR-EE}. In the Kaluza-Klein example there is also a physical interpretation. Equation \eqref{eq:intro-torus-pde-2} includes a term of the form
$$ \delta^\abbar m^{cd} \left( \frac{1}{4}  \p_a h_\abbar\p_b h_{cd} + \frac{1}{4}  \p_b h_\abbar \p_a h_{cd} \right) \,.$$
This describes interactions between the $(3+1)-$dimensional metric and the scalar fields $g_\abbar$. Similar terms exist in $Q_\mn$ and $G_\mn$ also. This non-trivial coupling cannot come from the stress-energy tensor of a massless scalar field, and so the more general PDE system \eqref{eq:intro-new-EE} is required. 

The method of Lindblad and Rodnianski has also been used by Choquet-Bruhat, Loizelet and Chrusciel to show the non-linear stability of Minkowski spacetime $\R^{1+n}$ for $n\geq 4$ \cite{ChoquetBruhat:2006jc}. Indeed following the method of \cite{LR:04}, Loizelet proved in \cite{MR2582443} that the Minkowski solution minimally coupled to the  Maxwell equations is non-linearly stable. The system considered was:
\begin{equation} \begin{split}
R_{ab} &= 2 \left( \mathcal{F}_{ac} \mathcal{F}_b {}^c - \frac{1}{4} g_{ab} \mathcal{F}^{cd} \mathcal{F}_{cd} \right) \,,\\
\nabla_a \mathcal{F}^{ab} &=0   \,. \end{split} \label{eq:intro-cb-loiz} \end{equation}
Using the wave-coordinate condition \eqref{eq:intro-wave-gauge} and the Lorenz gauge $\p_\mu \left( \sqrt{\det g} \A^\mu \right)=0$ this PDE reduces to a system of the form
\begin{equation} \begin{split}
\tbox_g h_\mn & = F_\mn (h)(\p h, \p h) + \widetilde{F}_\mn (h)(\p A, \p A) \,, \\
\tbox_g \mathcal{A}_\mu & = F^\A_\mu (h) (\p h, \p A) \,. \end{split} \label{eq:loiz-reduced} \end{equation} 
The PDE \eqref{eq:loiz-reduced} has very similar properties to our general PDE \eqref{eq:intro-new-EE}, see Remark \ref{remark:loizelet} for more details. Furthermore \eqref{eq:intro-motiv-ems} would reduce to  \eqref{eq:intro-cb-loiz} \textit{if} $\phi$ could be set to a constant, thus unifying gravity and electromagnetism from a pure higher-dimensional gravity.  However the equation of motion \eqref{eq:intro-motiv-ems-3} for $\phi$ is not a free wave equation, but rather has an inhomogeneity involving $\mathcal{F}$. From this one sees that it is inconsistent to simply take $\phi = 0$, since this Ansatz would not satisfy the equations of motion. Thus our more generalised system \eqref{eq:intro-new-EE} is needed to treat \eqref{eq:intro-motiv-ems}. This should be contrasted with our earlier Ansatz that $G^{(n)}_\mn = 0$ for all $n \neq 0$. This does not violate the equations of motion and thus gives what is often called a consistent truncation.

Of course not all Einstein systems fall into the class considered in \eqref{eq:intro-new-EE}. An example is general relativity coupled to non-linear electromagnetic fields, such as the Born-Infeld system, considered in \cite{Speck}. 

An obvious question not answered in our work is to consider all modes not just the $n=0$ mode.  Furthermore one could ask whether other product spacetimes of interest in string theory, such as $\R^{1+3} \times \mathcal{M}$ where $\mathcal{M}$ is a Calabi-Yau manifold, are non-linearly stable against some, or all, perturbations of initial data. In the full $n-$mode case the non-linear wave equations contain Klein-Gordon terms, which have only recently been treated at the full non-linear level by \cite{LeFloch:2015ppi}.

\subsection{The Weak Null Condition} \label{section:intro-weak-null}
As alluded to so far, there are important properties of a non-linearity which imply, or at least allow one to hope, that a PDE may be solved globally. One such condition on the non-linearities is the \textit{weak null condition}, first introduced by Lindblad and Rodnianski in \cite{lindblad:weaknull}. Their idea was to look at an asymptotic form of the PDE and determine whether this has global solutions. Here we just state the condition following the notation of the original work, but for more details and motivation see \cite{lindblad:weaknull, LR:04}. 

\begin{definition}
Consider a system of hyperbolic PDEs for unknowns $u_i, i \in \{1, \ldots, N \}$
\begin{equation} \begin{split}
 - \tbox_m u_i &= a^{jk}_{i \alpha \beta} \p^\alpha u_j \p^\beta u_k + G_i (u, \p u, \p^2 u) \,, \\
 u(t=0,x) &= \varepsilon f(x) \in C^\infty \,, \p_t u(t=0,x) = \varepsilon g (x) \in C^\infty  \,,
 \end{split} \label{eq:weak-null-pde}\end{equation}
where $G_i$ vanishes to third order as $(u, \p u, \p^2 u) \to 0$ and $a^{jk}_{i \alpha \beta}  = 0$ unless $\vert \alpha \vert \leq \vert \beta \vert \leq 2$ and $ \vert \beta \vert \geq 1$. Note also $\p^\alpha$ are multi-indices, not coordinate indices. 
Take the following asymptotic expansion 
\begin{equation} u_i(t,x) \sim \frac{\varepsilon U(q,s,\omega)}{r} \,, q := r-t, s:= \varepsilon \ln r, \omega := \vert x \vert / r \,, \label{eq:weak-null-exp} \end{equation}
as $r \to \infty$ and in the wave-zone region $r \sim t$.
Equating terms of order $\mathcal{O}(\varepsilon^2 r^{-2})$ gives the system
\begin{equation} \begin{split}
2 \p_s \p_q U_i &= A^{jk}_{i, mn} (\omega) \p_q ^m U_j \p^n U_k \\
U_i(s=0) &= F_0 \,, \label{eq:weak-null-asympt}
\end{split} \end{equation}
where
$$ A^{jk}_{i, m n}  (\omega) := \sum_{\vert \alpha \vert = m, \vert \beta \vert = n} a^{jk}_{i \alpha \beta}  \hat{\omega}^\alpha \hat{\omega}^\beta \,, \quad \hat{\omega} := (-1, \omega) \,, \quad \hat{\omega}^\alpha = \prod_{\sum \vert \alpha_i \vert := \vert \alpha \vert } \hat{\omega}_{\alpha_1} \cdots \hat{\omega}_{\alpha_k} \,.$$
The system \eqref{eq:weak-null-pde} is said to satisfy the weak null condition if the asymptotic system \eqref{eq:weak-null-asympt} has solutions for all $s$ and if the solutions $U_i$ together with their derivatives grow at most exponentially in $s$ for all initial data decaying sufficiently fast in $q$. 
\end{definition}

Note the null-forms \eqref{eq:intro-nullforms} satisfy a stricter condition called the \textit{null condition}, analysed by Klainerman in \cite{MR837683}, which requires $ A^{jk}_{i mn} (\omega) \equiv 0$.

At each point $(t,x) \in \R \times (\R^3 \backslash \{ 0 \} )$ introduce the null vectors $(L, \lbar)$ defined by
$$ L^0 = 1, L^i = x^i / \vert x \vert \,, \quad \lbar^0 = 1, \lbar^i = - x^i / \vert x \vert \,. $$
In a neighbourhood of each point, we can also find a pair of orthonormal vector fields $(S_1, S_2)$ orthogonal to the null pair $(L, \lbar)$. This set $\U := \{ L, \lbar, S_1, S_2 \}$ is the null frame, see also Section \ref{section:null-frame} for more details. Let the inverse perturbation be
$ H^\mn := g^\mn - m^\mn$.
As described in \eqref{eq:weak-null-exp}, we take the asymptotic expansion 
$$h_\mn \sim \frac{\varepsilon D_\mn}{r} \,, \quad \psi_k \sim \frac{\varepsilon V_k}{r} \,,$$ 
and substitute this into the PDE system \eqref{eq:intro-new-EE} to yield
\begin{subequations} \label{eq:intro-gen-weaknull}
 \begin{align}
(2 \p_s - H_{LL} \p_q ) \p_q D_\mn &= L_\mu L_\nu \left[ P^1(\p_q D, \p_q D) + P^2 (\p_q D, \p_q V) + P^3 (\p_q V, \p_q V) \right] \label{eq:intro-asymp1} \,, \\
(2 \p_s - H_{LL} \p_q ) \p_q V_k &= 0 \label{eq:intro-asymp2} \,, \\
2\p_q D_{L \mu} &= L_\mu \p_q (\tr_m D) \,. \label{eq:intro-asymp3}
 \end{align}
\end{subequations}
Here $P^i$ are non-null terms defined in \eqref{eq:full-P-general-PDE} and \eqref{eq:intro-asymp3} is the asymptotic form of the wave-coordinate condition \eqref{eq:intro-wave-gauge}. Following the argument of \cite{lindblad:weaknull}, by \eqref{eq:intro-asymp1}
$$ (2 \p_s - H_{LL} \p_q) 2 \p_q D_{L \mu} = 0 \,,$$
implying $\p_q D_{L \mu}$ is preserved along the integral curves of the vector field $(2 \p_s - H_{LL} \p_q)$ and hence \eqref{eq:intro-asymp3} is preserved under the flow of \eqref{eq:intro-asymp1}. Contract \eqref{eq:intro-asymp1} with $L^\mu L^\nu$ to obtain the equation 
$$ (2 \p_s - H_{LL} \p_q ) \p_q D_{LL} = 0 \,, $$
which can be solved globally for $D_{LL}$. Using this solution, one can now solve \eqref{eq:intro-asymp2} for $V_k$. Furthermore by contracting \eqref{eq:intro-asymp1} with $T^\mu U^\nu$, for $T \in \{ L, S_1, S_2 \}, U \in \{ L,\lbar, S_1, S_2 \}$ we obtain the equation 
$$ (2 \p_s - H_{LL} \p_q ) \p_q D_{TU} = 0 \,.$$
This can also be solved globally and thus the only remaining unknown component is $H_{\lbar \lbar}$. Contracting \eqref{eq:intro-asymp1}  with $\lbar^\mu \lbar^\nu$ gives
$$ (2 \p_s - H_{LL} \p_q ) \p_q D_{\lbar \lbar} = 4 \left[ P^1(\p_q D, \p_q D) + P^2 (\p_q D, \p_q V) + P^3 (\p_q V, \p_q V) \right] \,.$$
Since\footnote{See \cite{lindblad:weaknull} for further discussion here.} the RHS does not contain $(\p_q H_{\lbar \lbar})^2$ terms, the equation can be solved globally for $H_{\lbar \lbar}$. 
The above argument implies the following proposition:
\begin{proposition}
The asymptotic system for the PDE \eqref{eq:intro-new-EE} takes the form \eqref{eq:intro-gen-weaknull}. The solution for this system \eqref{eq:intro-gen-weaknull} exists globally with all components remaining uniformly bounded while $\p_q H_{\lbar \lbar}$ grows at most as $s$. Thus \eqref{eq:intro-new-EE} satisfies the weak null condition.  
\end{proposition}

Moreover Theorem \ref{theorem:intro1} could be seen as an example supporting the conjecture that non-linear wave equations arising from the Einstein equation and wave-coordinate condition that satisfy the weak null conjecture have global-in-time solutions for small data.

\section{The null frame} \label{section:null-frame}
Define the local pair of null vectors $L, \lbar$ as in \cite{LR:04} 
$$ L^0 = 1, L^i = x^i / \vert x \vert \,, \quad \lbar^0 = 1, \lbar^i = - x^i / \vert x \vert \,, $$
where $i = 1,2,3$. Note that $L$ is tangent to the outgoing Minkowski null cones $\{ (t, x) : \vert x \vert - t = q \}$ and $\lbar$ is tangent to the ingoing cones $\{ (t, y) : \vert y \vert + t = s \}$. 
Furthermore
$$ L = \p_t + \p_r \,, \lbar = \p_t - \p_r \,.$$
Locally let $S_1, S_2$ be orthonormal smooth vector fields spanning the tangent space of the spheres $S^2$. Then  $\mathcal{U}:= (L, \lbar, S_1, S_2)$ forms a null frame.  Define
$$ \p_s := \frac{1}{2}(\p_r + \p_r) \,, \p_q:= \frac{1}{2} ( \p_r - \p_t) \,.$$
Relative to the null frame $\mathcal{U}$ the Minkowski metric $m_\mn$ takes the form
$$ m_{LL} = m_{\lbar \lbar} = m_{LA} = m_{\lbar A} = 0 \,, m_{L \lbar} = m_{\lbar L} = -2 \,, m_{AB} = \delta_{AB} \,, $$
where $A, B$ denote any of the vectors $S_1$ and $S_2$. Since $S^2$ does not admit a global orthonormal frame,  we consider the projections of $S_1$ and $S_2$ defined by
$$ \slashed{\p}_i := \p_i - \omega_i \omega^j \p_j \,, \quad  \omega^i := x^i / \vert x \vert \,, \quad 
i=1,2,3 \,. $$
If we denote $\bar{\p}_i := \slashed{\p}_i$  then the set $\{L, \lbar, \bar{\p}_1, \bar{\p}_2, \bar{\p}_3 \}$ defines a global frame. Furthermore if we define $\bar{\p}_0 := L^\mu \p_\mu$ then the set 
$$\bar{\p}:= \{ \bar{\p}_0, \bar{\p}_1, \bar{\p}_2, \bar{\p}_3\} $$ spans the tangent space of the outgoing light cone. As in \cite{LR:04}, define the following notation:
\begin{definition} \label{def:U,T,L,seminorms}
Let
\begin{equation}
\mathcal{U} := \{ L, \lbar, S_1, S_2 \} \,, \mathcal{T} := \{ L, S_1, S_2 \} \,, \mathcal{L} := \{ L \} \,, \label{eq:def-U,T,L}
\end{equation}
where $\mathcal{T}$ is the set of null frame vector fields tangent to the outgoing cones.
For any two families $\mathcal{V}$ and $\mathcal{W}$ of vector fields and an arbitary 2-tensor $\pi$, define the following pointwise seminorms
\begin{align*}
\vert \pi \vert_{\mathcal{V} \mathcal{W}} & := \sum_{V \in \mathcal{V}, W \in \mathcal{W}} \vert \pi_\mn V^\mu W^\nu \vert \,, \\
\vert \p \pi \vert_{\mathcal{V} \mathcal{W}} & := \sum_{U \in \mathcal{U}, V \in \mathcal{V}, W \in \mathcal{W}} \vert ( \p_\rho \pi_{\mu \nu}) U^\rho V^\mu W^\nu \vert \,, \\
\vert \pgood \pi \vert_{\mathcal{V} \mathcal{W}} & := \sum_{T \in \mathcal{T}, V \in \mathcal{V}, W \in \mathcal{W}} \vert ( \p_\rho \pi_{\mu \nu}) T^\rho V^\mu W^\nu \vert \,.
\end{align*}
Furthermore for a collection $\{ \Psi_k \}_{k \in \K}$ where $\K := \{1 , \ldots, m \}$, let
\begin{align*}
\vert \Psi \vert_\K := \sum_{k=1}^m \vert e^k \Psi_k \vert \,,
\end{align*}
where $(e^k)$ is the standard basis on $\R^m$. 
\end{definition}

\section{The (extended) Einstein Equations and Wave Coordinates} \label{section:ee-wave-coords}
In this section we look at the structure of the non-linearity of the PDE \eqref{eq:intro-new-EE} with respect to the null frame. Recall our unknown variables are
$$W = \{ h_\mn \}_{\mu, \nu \in \{ 0, 1, 2, 3 \}} \cup \{ \psi_k \}_{k \in \K} \,,$$ 
where $ h_\mn = g_\mn - m_\mn$ is the perturbation from the background Minkowski. For small $h$, the inverse perturbation is
$$ H^\mn := g^\mn - m^\mn = -h^\mn + \mathcal{O}^\mn (h^2) \,,$$
where $h^\mn := m^{\mu \rho} m^{\nu \sigma} h_{\rho \sigma}$ and $\mathcal{O}^\mn (h^2)$ vanishes to second order at $h=0$. The PDE, repeated from \eqref{eq:intro-new-EE}, is
\begin{subequations} \label{eq:general-pde}
\begin{equation} \begin{split}
\tbox_g h_\mn &= F_\mn (W) (\p W, \p W) \,, \\
\tbox_g \psi_k &= F_k (W) (\p W, \p W) \,, \\
\end{split} 
\end{equation}
where we defined the non-linearities by
\begin{equation} \begin{split}
F_\mn (W) (\p W, \p W) &= P(\p_\mu W, \p_\nu W) + Q_\mn (\p W, \p W) + G_\mn(W) (\p W, \p W) \,, \\
F_k (W) (\p W, \p W) &= Q_k (\p W, \p W) + G_k ( W) (\p W, \p W) \,,
\end{split} \label{eq:general-pde-F}
\end{equation}
together with the wave-coordinate condition
\begin{equation}
\p_\rho \left( g^{\rho \mu} \sqrt{\vert \det g \vert} \right) = 0 \,. \label{eq:general-pde-gauge}
\end{equation}
The quadratic terms $Q_\mn, Q_k$ are linear combinations of the null forms \eqref{eq:intro-nullforms} in terms of $\p W$ variables, contracting with $m_\mn$ and/or arbitrary $N^k \in \R^m$ as appropriate. While the remaining non-null $\mathcal{O}((\p W)^2)$ terms are of the form
\begin{equation} \begin{split}
P(\p_\mu W, \p_\nu W) & := P^1( \p h, \p h)_\mn + P^2 (\p h, \p \psi)_\mn + P^3 (\p \psi, \p \psi) _\mn \,, \\
P^1(\p h, \p h) _\mn & :=  m^{\rho \sigma} m^{\lambda \tau} \left( \frac{1}{4}  \p_\mu h_{\rho \sigma} \p_\nu h_{\lambda \tau} - \frac{1}{2}  \p_\mu h_{\rho \lambda} \p_\nu h_{\sigma \tau} \right) \,, \\
P^2 (\p h, \p \psi)_\mn &:= N^k m^{\rho \sigma} \left( \p_{\mu} h_{\rho \sigma}  \p_{\nu} \psi_k +  \p_{\nu} h_{\rho \sigma}  \p_{\mu} \psi_k  \right) \,, \\
P^3 (\p \psi, \p \psi)_\mn &:= N^{kl} \left( \p_{\mu} \psi_k \p_{\nu} \psi_l + \p_{\nu} \psi_k \p_{\mu} \psi_l \right) \,.
\end{split} \label{eq:full-P-general-PDE} \end{equation}
\end{subequations} 
Again $N^k \in \R^m, N^{kl} \in \R^{m^2}$ are some arbitrary coefficients.
Lastly the terms $G_\mn (W)(\p W, \p W)$ and $G_k (W)(\p W, \p W)$ are quadratic forms in $\p W$, with coefficients depending smoothly on $W$ and vanishing for $W=0$. Note we have used compact notation in \eqref{eq:general-pde}, writing $P( \p_\mu W, \p_\nu W)$ to represent $ P(\p h, \p h, \p \psi, \p \psi)_\mn$, similarly for $F, Q$ and $G$ also. 

\begin{proposition} \label{remark:torus-full-pde}
The $n=0$ mode reduction of Kaluza Klein on a $\tee$ subject to the wave-coordinate condition \eqref{eq:intro-wave-gauge} is included in the generalised PDE system \eqref{eq:general-pde}.
\end{proposition}

\begin{proof}
Under \eqref{eq:TIP} and the wave-coordinate condition the vacuum Einstein equations in $\R^{1+3} \times \tee$ reduce to the system
\begin{subequations}
\begin{align}
\tbox_g h_{ab} &= P(\p_a h, \p_b h)  + Q_{ab} (\p h, \p h) + G_{ab} (h) (\p h, \p h) \,, \label{eq:torus-pde-1}\\
\tbox_g h_{a \bbar} &= Q_{a \bbar} (\p h, \p h) + G_{a \bbar} (h) (\p h, \p h)\,, \\
\tbox_g h_\abbar &= Q_\abbar (\p h, \p h) + G_\abbar (h) (\p h, \p h) \,.
\end{align} \label{eq:torus-pde}
\end{subequations}
The $\mathcal{O}((\p h)^2)$ non-linearities are given by 
\begin{align*}
P(\p_a h, \p_b h) &= P^1(\p_a h, \p_b h)+P^2(\p_a h, \p_b h)+P^3(\p_a h, \p_b h) \,, \\
P^1 (\p_a h, \p_b h) &  =  m^{a'b'} m^{cd} \left( \frac{1}{4} \p_a h_{a'b'} \p_\nu h_{cd} - \frac{1}{2} \p_a h_{a'c} \p_b h_{b'd} \right) \,, \\
\end{align*}
\begin{align*}
P^2 (\p_a h, \p_b h)& = \delta^\abbar m^{cd} \left( \frac{1}{4}  \p_a h_\abbar\p_b h_{cd} + \frac{1}{4}  \p_b h_\abbar \p_a h_{cd}   \right) \,, \\
P^3 (\p_a h, \p_b h)& = - \delta^\abbar m^{cd} \left( \p_a h_{\abar c} \p_b h_{\bbar d} \right) + \delta^\abbar \delta^{\cbar \, \dbar} \left( \frac{1}{4} \p_a h_\abbar \p_b h_{\cbar \, \dbar} - \frac{1}{2}  \p_a h_{\abar \, \cbar} \p_b h_{\bbar \, \dbar} \right) \,,\\
Q_{ab}(\p h, \p h) &=   m^{c c'} Q_0 \left( h_{ca}, h_{c'b} \right) + \delta^{\cbar \cbar'} Q_0 \left( h_{\cbar a}, h_{\cbar 'b} \right) \\
& \quad + m^{dd'} \left( m^{cc'} Q_{dc'} (h_{ca}, h_{d'b}) + \delta^{\cbar \cbar'} Q_{ad} (h_{d' \cbar'} , h_{\cbar b} ) \right) + (a \leftrightarrow b) \\
& \quad + \frac{1}{2} m^{c c'} \left( m^{dd'} Q_{c' a} (h_{dd'}, h_{cb} ) + \delta^{\dbar \dbar'} Q_{c' a} ( h_{\dbar \dbar'}, h_{cb} ) \right) + (a \leftrightarrow b ) \,, \\
Q_{a \bbar}(\p h, \p h) &=  \delta^{\cbar \, \dbar} Q_0 ( h_{\cbar a}, h_{\dbar \bbar}) +  m^{cb} \delta^{\cbar \, \dbar} \left( Q_{ac}(h_{b\dbar}, h_{\cbar \bbar}) + \frac{1}{2} Q_{ba} (h_{\cbar \dbar}, h_{c \bbar}) \right) \,, \\
Q_{\abar \bbar}(\p h, \p h) &=  \delta^{\cbar \, \dbar} Q_0 (h_{\cbar \abar}, h_{\dbar \bbar} )\,.
\end{align*}
Here $(a \leftrightarrow b)$ indicates the previous bracketed term are repeated with $a$ and $b$ swapped. 
Recall also the indices $a,b$ and $\abar, \bbar$ are as defined in Section \ref{section:intro-motiv-eg}. Clearly $Q_{ab}$ and $Q_\abbar$ are linear combinations of the standard null forms \eqref{eq:intro-nullforms} and so satisfy the required estimates. If we consider the set $\{ \psi_k \} := \{ h_{a \bbar}, h_\abbar \}$ then we see the PDE \eqref{eq:torus-pde} falls into the system described in \eqref{eq:general-pde}
\end{proof}

\begin{remark}
As illustrated in the $d=1$ case in \eqref{eq:intro-motiv-ems}, the Einstein equation for the higher-dimensional metric can be interpreted from the view of the lower-dimensional theory as a $(3+1)-$dimensional metric coupled to a Maxwell vector potential and a scalar field. 

This structure is reflected in \eqref{eq:torus-pde-1}, which comes from the non-vacuum $(3+1)-$dimensional Einstein equation. Up to redefining variables, as done in \eqref{eq:intro-motiv-ansatz}, one can heuristically see the coupling between the Maxwell fields $h_{a \bbar}$, scalar fields $h_\abbar$ and the geometry $h_{ab}$. 

The first term $P^1(\p_a h, \p_b h)$ contains the standard non-null expression in \eqref{eq:intro-LR-EE} studied by Lindblad and Rodnianski coming from the self-interaction of the metric. 
The term $P^2(\p_a h, \p_b h)$ comes from the non-trivial coupling between the background geometry $h_{ab}$ and the scalar fields $h_\abbar$. 
The first term in $P^3(\p_a h, \p_b h)$ involves the interactions between the vector potentials, while the second term comes from the stress-energy tensor $T^\psi$ for a scalar wave. 
\end{remark}

Next we turn to the structure of the non-linearity of the PDE \eqref{eq:general-pde} with respect to the null frame.
As in \cite{LR:04}, first express the non-linearity in terms of some general tensors and/or scalars. For example, for symmetric 2-tensors $p, k$ we define
$$ P^1(p, k) := m^{\rho \sigma} m^{\lambda \tau} \left( \frac{1}{4}  p_{\rho \sigma} k_{\lambda \tau} - \frac{1}{2}  p_{\rho \lambda} k_{\sigma \tau} \right) \,.$$
Note there are no free indices here. Similar definitions holds for $P^2, P^3$ with $\Phi, \Psi$ some arbitrary functions. All together, let
$$ P(p, k, \Psi, \Phi) := P^1(p, k) + P^2 (p, \Psi) + P^3(\Psi, \Phi) \,.$$
The reason for using this notation is that eventually we will want to calculate $Z^I F_\mn$ where $Z \in \mathcal{Z}$. This notation allows us to derive estimates which still hold even when we have distributed the $Z^I$ derivatives across the terms in the non-linearity. See for example Corollary \ref{corol-9.7}. 

For the null terms, it is irrelevant which components of the unknowns are being considered, but it is crucial which derivatives appear. Thus for some $\Pi, \Theta$, which are either 2-tensors or scalars as required, define
\begin{equation} \vert Q(\p \Pi, \p \Theta) \vert := \vert Q_\mn (\p \Pi, \p \Theta) \vert_\UU + \vert Q_k ( \p \Pi, \p \Theta) \vert_\K \,. \label{eq:null-norm} \end{equation}

\begin{lemma}[Modified Lemma 4.2 from \cite{LR:04}] \label{lemma-4.2}
Let $\pi, \theta$ be arbitrary 2-tensors, $\Phi, \Psi$ arbitrary functions and $\Pi, \Theta$ 2-tensors or scalars as required. For the $\mathcal{O}((\p W)^2)$ terms contained in \eqref{eq:full-P-general-PDE} we obtain
\begin{align*} 
\vert P(p, k, \Psi, \Phi)\vert & \lesssim \vert p \vert_\TU \vert k \vert_\TU + \vert p \vert_\LL \vert k \vert_\UU + \vert p \vert_\UU \vert k \vert_\LL + \vert p \vert_\TU \vert \Psi \vert_\K + \vert \Psi \vert_\K \vert \Phi \vert_\K \,, \\
\vert Q(\p \Pi, \p \Theta) \vert & \lesssim \vert \pgood \Pi \vert \vert \p \Theta \vert + \vert \p \Pi \vert \vert \pgood \Theta \vert \,.
\end{align*}
\end{lemma}
\begin{proof}
Expanding with respect to the null frame we find
\begin{align*}
\vert P^1 (p, k)\vert & \lesssim  \vert p \vert_\TU \vert k \vert_\TU + \vert p \vert_\LL \vert k \vert_\UU  + \vert p \vert_\UU \vert k \vert_\LL \,, \\ 
\vert P^2 (p, \Psi) \vert  &\lesssim \vert p \vert_\TU \vert \Psi \vert_\K \,, \\
\vert P^3 (\Psi, \Phi) \vert  &\lesssim \vert \Psi \vert_\K \vert \Phi \vert_\K \,.
\end{align*}
Note that $P^2(\p h, \p \psi)_\mn$ involves $m^{\rho \sigma} h_{\rho \sigma} = \tr_m h$. By considering possible indices, we see another possible term is of the form 
$$ N^k m^{\rho \sigma} \left( \p_\rho h_{\mu \sigma} \p_\nu \psi_k + \p_\rho h_{\nu \sigma} \p_\mu \psi_k \right) \,.$$
However this term does not contain $\tr_m h$ and so in $\vert P^2 (p, \Psi) \vert$ we would pick up additional terms of the form $\vert p \vert_\UU \vert \Psi \vert_\K$ which cannot be controlled. 
The estimate for $\vert Q \vert$ comes from the definition of $Q_\mn$ and $Q_k$ in \eqref{eq:general-pde} and estimate \eqref{eq:intro-null-estimate} for null forms.  
\end{proof}

\begin{remark}
Our proof follows very closely the proof of \cite{LR:04}. However to be self-contained, this paper repeats several results directly from \cite{LR:04}. Where a result here is given as `modified' from a result in \cite{LR:04}, it generally follows from their proofs with at most a very simple calculation difference.
\end{remark}

\section{Wave Gauge and Estimates of the Inhomogeneity} \label{section:wave-gauge} 
Recall the definition of the wave-coordinate condition
\begin{equation}
\p_\rho \left( g^{\rho \nu} \sqrt{\vert \det g \vert} \right) = 0 \,. \label{eq:wave-gauge}
\end{equation}
The following Lemmas exploit to great effect this choice of gauge and properties of the null frame. In particular Lemma \ref{lemma-8.1} implies that we can exchange a full derivative $\p$ on certain `good' components of the metric, with a good derivative $\pgood$ on all components of the metric, plus some higher order terms. 
\begin{lemma}[Lemma 8.1 from \cite{LR:04}]  \label{lemma-8.1}
Assume $g$ satisfies the wave-coordinate condition \eqref{eq:wave-gauge} and that the inverse perturbation satisifies $\vert H \vert_\UU \leq \frac{1}{4}$,  then
\begin{equation} \vert \p H \vert_\LT \leq \vert \pgood H \vert_\UU+ \vert H \vert_\UU \vert \p H \vert_\UU \,. \label{eq:wave-gauge-estimate} \end{equation} 
\end{lemma}
\begin{remark}
In the $n=0$ mode reduction of Kaluza-Klein on a $\tee$, the above estimate \eqref{eq:wave-gauge-estimate} holds with $\mathcal{T}$ replaced with $ \mathcal{T} \cup \{ K_1, \ldots, K_m \}$ and $\U $ replaced with $ U \cup \{ K_1, \ldots, K_m \}$ where $\{ K_1, \ldots, K_m \}$ is the collection of vector fields spanning the $\tee$.
\end{remark}
One can also commute through vector fields $Z \in \mathcal{Z}$ in \eqref{eq:wave-gauge-estimate} to estimate `good' components $ \vert \p Z^I H \vert_\TU$. 
\begin{proposition}[Proposition 8.2 from \cite{LR:04}] \label{prop-8.2}
Suppose that $g$ satisfies the wave-coordinate condition \eqref{eq:wave-gauge} and that for some $\I$ we have $\vert Z^J H \vert \leq C$ for all $\vert J \vert \leq \vert I \vert / 2$ and for all $Z \in \mathcal{Z}$. Then
\begin{align*}
\vert \p Z^I H \vert_\LT & \lesssim \sum_{\J\leq \I} \vert \pgood Z^J H \vert_\UU  + \sum_{\J \leq \I-1 } \vert \p Z^J H \vert _\UU + \sum_{\vert J_1 \vert + \vert J_2 \vert \leq \I} \vert Z^{J_1} H \vert_\UU  \vert \p Z^{J_2} H \vert_\UU \,,\\
\vert \p Z^I H \vert_\LL & \lesssim \sum_{\J \leq \I} \vert \pgood Z^J H \vert_\UU  + \sum_{\J \leq \I-2 } \vert \p Z^J H \vert _\UU + \sum_{\vert J_1 \vert + \vert J_2 \vert \leq \I} \vert Z^{J_1} H \vert_\UU  \vert \p Z^{J_2} H \vert_\UU \,.
\end{align*}
\end{proposition}

Furthermore to keep track of the derivatives as well as the components, introduce the following notation for arbitrary 2-tensors $\pi, \theta$ and functions $\Psi, \Phi$
$$ P(\p \pi, \p \theta, \p \Psi, \p \Phi)_\mn := P^1(\p \pi, \p \theta)_\mn + P^2 (\p \pi, \p \Psi)_\mn + P^3(\p \Psi, \p \Phi) _\mn\,,$$
where for example
$$ P^1(\p \pi, \p \theta)_\mn := m^{\rho \sigma} m^{\lambda \tau} \left( \frac{1}{4}  \p_\mu \pi_{\rho \sigma} \p_\nu \theta_{\lambda \tau} - \frac{1}{2}  \p_\mu \pi_{\rho \lambda} \p_\nu \theta_{\sigma \tau} \right) \,,$$
and similarly for $P^2$ and $P^3$ as defined in \eqref{eq:full-P-general-PDE}. 
Using Lemma \ref{lemma-8.1} and Proposition \ref{prop-8.2} we now aim to estimate the non-linearities in \eqref{eq:general-pde-F}. 
\begin{lemma}[Modified Lemma 9.6 from \cite{LR:04}] \label{lemma-9.6}
The quadratic form $P_\mn$ defined in \eqref{eq:full-P-general-PDE} satisfies the following estimates
\begin{align*} \vert P (\p \pi, \p \theta, \p \Psi, \p \Phi) \vert_\TU & \lesssim \vert \pgood \pi \vert_\UU \vert \p \theta \vert_\UU + \vert \p \pi \vert_\UU \vert \pgood \theta \vert_\UU + \vert \pgood \pi \vert_\TU \vert \p \Psi \vert_\K \\
& \quad + \vert \p \pi \vert_\TU \vert \pgood \Psi \vert_\K + \vert \pgood \Psi \vert_\K \vert \p \Phi \vert_\K + \vert \p \Psi \vert_\K \vert \pgood \Phi \vert_\K \,, \\
\vert P (\p \pi, \p \theta, \p \Psi, \p \Phi) \vert_\UU &\lesssim \vert \p \pi \vert_\TU \vert \p \theta \vert_\TU + \vert \p \pi \vert_\LL \vert \p \theta \vert_\UU + \vert \p \pi \vert_\UU \vert \p \theta \vert_\LL \\
& \quad + \vert \p \pi \vert_\TU \vert \p \Psi \vert_\K + \vert \p \Psi \vert_\K \vert \p \Phi \vert_\K \,.
\end{align*}
\end{lemma}
\begin{proof}
The first estimate follows by contracting \eqref{eq:full-P-general-PDE} with $T^\mu S^\nu$ for $T \in \mathcal{T}, S \in \mathcal{U}$.
For the second estimate, replace $p\,, \Psi \,, k \,, \Phi$ from Lemma \ref{lemma-4.2} with $U^\mu \p_\mu p \,, U^\mu \p_\mu \Psi \,, S^\nu \p_\nu k \,, S^\nu \p_\nu \Phi $ respectively. For $S, U \in \mathcal{U}$ one obtains
\begin{align*}
\vert S^\mu U^\nu  P(\p \pi, \p \theta, \p \Psi, \p \Phi)_\mn \vert & \lesssim \vert U^\mu \p_\mu \pi \vert_\TU \vert S^\nu \p_\nu \theta \vert_\TU + \vert U^\mu \p_\mu \pi \vert_\LL \vert S^\nu \p_\nu \theta \vert_\UU + \vert U^\mu \p_\mu \pi \vert_\UU \vert S^\nu \p_\nu \theta \vert_\LL \\
& + \vert U^\mu \p_\mu \pi \vert_\TU \vert S^\nu \p_\nu \Psi \vert_\K + \vert U^\mu \p_\mu \Psi \vert_\K \vert S^\nu \p_\nu \Phi \vert_\K \,.
\end{align*}
Lastly note that $\vert \pgood \phi \vert = \sum_{i=0}^3 \vert \pgood_i \phi \vert$ is equivalent to $\sum_{T \in \mathcal{T}} \vert T^\mu \p_\mu \phi \vert$, similarly for $\vert \p \phi \vert$ and $\sum_{U \in \mathcal{U}} \vert U^\mu \p_\mu \phi \vert$. 
\end{proof}
It turns out the wave-coordinate condition will imply a hierarchy of components, with certain `good' variables $W_G$ having better decay rates than the full set $W$. The following notation is meant to simplify the notation and distinguish between components with different decay rates. 
\begin{definition}
Let 
\begin{equation} \begin{split}
W &= \{ U^\mu V^\nu h_\mn , \psi_k : U, V \in \mathcal{U} \,, k \in \K \} \,, \\
\quad W_G & := \{ U^\mu T^\nu h_\mn, \psi_k : U \in \mathcal{U} \,, T \in \mathcal{T} \,, k \in \K \}  \,.
\end{split} \label{eq:def-W-Y} \end{equation}
In particular we define 
$$ \vert W \vert := \vert h \vert_\UU + \vert \psi \vert_\K \,, \quad \vert W \vert_G := \vert h \vert_\TU + \vert \psi \vert_\K \,. $$
For derivatives we define 
$$ \vert \p W \vert := \vert \p h \vert_\UU + \vert \p \psi \vert_\K \,, \quad \vert \pgood W \vert := \vert \pgood h \vert_\UU + \vert \pgood \psi \vert_\K \,.$$
and analogous definitions for $\vert \p W \vert_G$ and $\vert \pgood W \vert_G$. 
\end{definition}
The norms without subscripts, $\vert \cdot \vert$ and \eqref{eq:null-norm}, indicate a sum over all the possible components of the $\UU$ and $\K$ terms. The subscript $\vert \cdot \vert_G$ indicates only the `good' components are being considered. This means $|W_G| = |W|_G$, however we generally use the latter throughout. 

Eventually the inhomogeneity to be studied will come from commuting $Z^I$ through the PDE \eqref{eq:general-pde}. Thus the following results estimate the non-linearities $F_\mn, F_k$ as well as $Z^I F_\mn, Z^I F_k$. 
\begin{corollary}[Modified Corollary 9.7 from \cite{LR:04}] \label{corol-9.7}
Assume $g$ satisfies the wave-coordinate condition \eqref{eq:wave-gauge} and $W, W_G$ are defined as in \eqref{eq:def-W-Y}, then
\begin{equation} \begin{split} 
\vert P (\p h, \p h, \p \psi, \p \psi ) \vert_\TU & \lesssim \vert \pgood W \vert \vert \p W \vert \,, \\
\vert P (\p h, \p h, \p \psi, \p \psi) \vert_\UU &\lesssim \vert \p W \vert_G^2 + \vert \pgood W \vert \vert \p W \vert + \vert W \vert \vert \p W \vert^2 \,. \label{eq:cor-9.7-P}
\end{split} \end{equation}
Furthermore, assuming that $\vert Z^J W \vert \leq C$ for all $\J \leq \I$ and for all $Z \in \mathcal{Z}$, we have
\begin{align*}
& \vert Z^I P (\p h, \p h) \vert_\UU
\lesssim  \sum_{\J + |K| \leq \I} \vert \p Z^J W \vert_G \vert \p Z^K W \vert_G  + \vert \pgood Z^J W \vert \vert \p Z^K W \vert \\
&  + \sum_{\J + |K|  \leq \I} \vert \p Z^K h \vert_\UU \left[ \sum_{\vert J'\vert \leq \vert J\vert -1} \vert \p Z^{J'} h \vert_\LT + \sum_{\vert J'' \vert \leq \J -2 } \vert \p Z^{J''} h \vert + \sum_{\vert J_1 \vert + \vert J_2 \vert \leq \J} \vert Z^{J_2} h \vert \vert \p Z^{J_1} h \vert \right] \,.
\end{align*}
It also holds that
\begin{align*}
\vert Q( \p W, \p W) \vert \lesssim \vert \pgood W \vert \vert \p W \vert  \,.
\end{align*}
\end{corollary}
\begin{proof}
The first identities \eqref{eq:cor-9.7-P} follow from Lemma \ref{lemma-9.6} and Lemma \ref{lemma-8.1}. For the estimate of $Z^I P$ we use Lemma \ref{lemma-9.6}, Proposition \ref{prop-8.2} and note that
$$ \vert Z^I P (\p h, \p h, \p \psi, \p \psi ) \vert_\UU \lesssim \sum_{| J_1| + \cdots+| J_4| \leq \I} \vert P ( \p Z^{J_1} h, \p Z^{J_2} h, \p Z^{J_3} \psi, \p Z^{J_4} \psi ) \vert \,. $$
Lastly Lemma \ref{lemma-4.2} implies the estimate for $\vert Q \vert$. 
\end{proof}

We can now put these results together in Proposition \ref{prop-9.8} to estimate the inhomogeneous terms $F_\mn$ and $F_k$. 
\begin{proposition}[Modified Proposition 9.8 from \cite{LR:04}] \label{prop-9.8}
Assume $g$ satisfies the wave-coordinate condition \eqref{eq:wave-gauge},  $W$ and $W_G$ are defined as in \eqref{eq:def-W-Y} and $F_\mn$ and $F_k$ be as defined in \eqref{eq:general-pde-F}. Then
\begin{equation} \begin{split}
\vert F \vert_\K + \vert F \vert_\TU  &\lesssim \vert \pgood W \vert \vert \p W \vert + \vert W \vert \vert \p W \vert^2 \,, \\
\vert F \vert_\UU & \lesssim \vert \p W \vert_G^2 + \vert \pgood W \vert \vert \p W \vert + \vert W \vert \vert \p W \vert^2  \,. \end{split} \label{eq:prop9.8-estimates}
\end{equation}
Furthermore if $\vert Z^J W \vert \leq C$ for all $\vert J \vert \leq \vert I \vert $ and for all $Z \in \mathcal{Z}$, then 
\begin{align*}
\vert Z^I F \vert_\UU & \lesssim  \sum_{\vert J \vert + \vert K \vert \leq \vert I \vert} \Big[ \vert \p Z^J W \vert_G \vert \p Z^K W \vert_G  + \vert \pgood Z^J W \vert \vert \p Z^K W \vert \Big] \\
&  + \sum_{\vert J \vert + \vert K \vert \leq \vert I \vert-2 } \vert \p Z^{J} h \vert_\UU \vert \p Z^K h \vert_\UU + \sum_{\vert J_1 \vert +\vert J_2 \vert + \vert J_3 \vert \leq \vert I \vert} \vert Z^{J_3} W \vert \vert \p Z^{J_2} W \vert \vert \p Z^{J_1} W \vert \,, \\
\vert Z^I F \vert_\K & \lesssim \sum_{\vert J \vert + \vert K \vert \leq \vert I \vert} \vert \pgood Z^J W \vert \vert \p Z^K W \vert + \sum_{\vert J_1 \vert +\vert J_2 \vert + \vert J_3 \vert \leq \vert I \vert} \vert Z^{J_3} W \vert \vert \p Z^{J_2} W \vert \vert \p Z^{J_1} W \vert \,.
\end{align*}
\end{proposition}
\begin{proof}
The required estimates come from Lemma \ref{lemma-9.6}, Corollary \ref{corol-9.7} and noting that
$$ \vert Z^I Q (\p W, \p W) \vert \lesssim \sum_{| J_1| + \vert J_2 \vert \leq \I} \vert Q ( \p Z^{J_1} W, \p Z^{J_2} W) \vert $$
and
$$ \vert Z^I G (W)(\p W, \p W) \vert \lesssim \sum_{| J_1| + \cdots| J_3| \leq \I} \vert Z^{J_1} W \vert \vert \p Z^{J_2} W \vert \vert \p Z^{J_3} W \vert \,.$$
\end{proof}

\begin{remark} \label{remark:loizelet}
\normalfont
Having introduced a lot of useful notation and obtained the first estimates of the inhomogeneities in Proposition \ref{prop-9.8}, we pause now to discuss some further points about the generalised inhomogeneities considered in the PDE \eqref{eq:general-pde}. Our system is of the form
\begin{equation} \label{remark:general-pde}
\begin{split}
\tbox_g h_\mn = F_\mn \,, \\
\tbox_g \psi_k = F_k \,.\end{split}
\end{equation}
Prescribing $F_k$ to have no non-null $\mathcal{O}((\p W)^2)$ inhomogeneities meant that  $\vert F \vert_\K$ has the same estimates as $\vert F \vert_\TU$, as seen in \eqref{eq:prop9.8-estimates}. It was also important that $F_\mn$ did not contain terms of the form $\p_\mu h_{\lbar \, \lbar} \p_\nu h_{\lbar \, \lbar}$ or $\p_\mu h_{\lbar \, \lbar} \p_\nu \psi_k$. Explicitly, the $P^1$ and $P^2$ terms, as shown in Lemma \ref{lemma-4.2}, did not contain terms estimated by $\vert \p h \vert_\UU \vert \p h \vert_\UU$ and $\vert \p h \vert_\UU \vert \p \psi \vert_\K$, but rather contained terms estimated by $\vert \p h \vert _\UU \vert \p h \vert_\LL$ and $\vert \p h \vert_\TU \vert \p \psi \vert_\K$ respectively. Having these `problem' terms would, apart from not allowing the energy argument to close, violate the weak null condition. 
 
 One should be wary that the system \eqref{remark:general-pde} is not the most general possible. However apart from the result of \cite{MR2582443}, it is unclear what generalisation would be the most fruitful at capturing other physically interesting scenarios. Indeed we could imagine including non-null $\mathcal{O}((\p W)^2)$ terms to $F_k$, provided that there are some good properties and additional assumptions implying that $F_k$ still obeys the same estimates as $\vert F \vert_\TU$. An example of where such a generalisation occurs is for the minimally coupled Einstein-Maxwell system considered in \cite{MR2582443}. This system is of the form
\begin{equation} \begin{split}
\tbox_g h_\mn & = F_\mn (h)(\p h, \p h) + \widetilde{F}_\mn (h)(\p A, \p A) \,, \\
\tbox_g \mathcal{A}_\mu & = F^\A_\mu (h) (\p h, \p A) \,, \end{split} \label{remark-loizelet-pde} \end{equation}
where $F_\mn(h)(\p h, \p h)$ is the original inhomogeneity \eqref{eq:intro-LR-EE} from \cite{LR:04}.  This system can be thought of in terms of our general PDE system \eqref{eq:general-pde} if we let $\{ \psi_k \}:= \{ \mathcal{A}_\mu \}$ and make some minor generalisations. Most significantly, one should include the `additional assumption' of the Lorenz gauge in \eqref{eq:general-pde-gauge}. Since  now $\K := \U$, additional terms can be added to $P^3$ of \eqref{eq:full-P-general-PDE}, namely swapping $\mu, \nu, k$ and $l$ indices and summing as needed. This would capture the required additional terms appearing in $\widetilde{F}_\mn $ of \eqref{remark-loizelet-pde}. Note since the coupling is minimal there are no $P^2$ terms appearing in $\widetilde{F}_\mn $ and so we need not worry about terms of the form $\vert \p h \vert_\UU \vert \p \A \vert_\U$ discussed in the previous paragraph. 

Now $F^\A_\mu$ contains terms quadratic in $\p \A$ which are not strictly null-forms and hence are not included in \eqref{remark:general-pde}. However by expanding out the Lorenz gauge $\p_\mu \left( \sqrt{\det g} \A^\mu \right)=0$ in the null frame one can obtain estimates on $\vert \p A \vert_{\mathcal{L}}$ similar in spirit to the wave gauge estimates \eqref{lemma-8.1}.  Using both gauges, and as shown in \cite{MR2582443}, the following estimates are obtained
\begin{align*}
\vert F \vert_\UU & \lesssim \vert \p h \vert^2_\TU + \vert \pgood h \vert_\UU \vert \p h \vert_\UU + \vert h \vert \vert \p h \vert^2 \,, 
& \vert F \vert_\TU & \lesssim \vert \pgood h \vert_\UU \vert \p h \vert_\UU + \vert h \vert \vert \p h \vert^2 \,, \\
\vert \widetilde{F} \vert_\UU & \lesssim \vert \p \A \vert_\U^2 + \vert h \vert \vert \p \A \vert^2 \,, 
& \vert \widetilde{F} \vert_\TU & \lesssim \vert \pgood \A \vert_\U \vert \p \A \vert_\U + \vert h \vert \vert \p \A \vert^2 \,, \\
\vert F^\A \vert_\U & \lesssim \vert \pgood h \vert_\UU \vert \p \A \vert + \vert \pgood \A \vert \vert \p h \vert_\UU + \vert h \vert \vert \p h \vert \vert \p \A \vert \,.
\end{align*}
Using our notation defined in \eqref{eq:def-W-Y}, these estimates become
\begin{align*}
\vert F^\A \vert_\K + \vert F +\widetilde{F} \vert_\TU & \lesssim \vert \pgood W \vert \vert \p W \vert + \vert W \vert \vert \p W \vert^2 \,, \\
\vert F + \widetilde{F} \vert_\UU &\lesssim \vert \p W \vert_G^2 + \vert \pgood W \vert \vert \p W \vert + \vert W \vert \vert \p W \vert^2 \,.  
\end{align*}
Thus the estimates in \cite{MR2582443} agree with those obtained in Proposition \ref{prop-9.8}. 
\end{remark}

\section{Decay Estimates -- Part I} \label{section:decay-1} 
The estimates in this section fall into two parts. The first section involves decay estimates for solutions $\phi$ of $\tbox_g \phi = F$ where $F$ will eventually be of the form in \eqref{eq:general-pde-F}. The second part involves `weak decay' estimates coming from the Klainerman-Sobolev inequality and a bootstrap assumption on our energy. 

For the first part, we will in fact require a weighted $L^\infty$ estimate of the solution $\phi$, and so we define this \textit{decay weight} as
\begin{equation} \varpi := \varpi(q) = \Bigg\lbrace \begin{array}{ll}
(1+\qv)^{1+\gamma'} \,, & q>0 \\
(1+\qv)^{1/2-\mu'} \,, & q<0 \,,
\end{array} \label{eq:def-varpi} \end{equation}
where $ \gamma ' \geq -1$ and $\mu' \leq 1/2$ are some constants and $q=t-r$. The details can be found in \cite{LR:04}, however the proof of the following Corollary \eqref{corol-7.2} essentially comes from applying the fundamental theorem of calculus along the integral curves of the vector field $\p_s + \frac{H^{\lbar \lbar} }{2g^{L \lbar}} \p_q$. This vector field can be interpreted as the  vector field $\p_s$ in the perturbed geometry, shifted by some small amount towards $\p_q$. 

\begin{corollary}[Corollary 7.2 from \cite{LR:04}] \label{corol-7.2}
Let $\phi_\mn$ be a solution of the system
$$ \tbox_g \phi_\mn = F_\mn \,,$$
where $F_\mn$ is some as yet unspecified non-linearity. 
Assume that $H^\mn = g^\mn - m^\mn$ satisfies
\begin{equation}
\vert H \vert_\UU \leq \varepsilon' \,, \quad \int_0^\infty \left\Vert \vert H (t,\cdot) \vert_\UU \right\Vert_{L^\infty (D_t)} \frac{dt}{1+t} \leq \frac{1}{4} \,, \quad \vert H \vert_\LT \leq \varepsilon' \frac{\qv +1}{1+t+\vert x \vert} \,,
\end{equation}
in the region $D_t := \{ x \in \R^3 : t/2 \leq \vert x \vert \leq 2t \}$ for some $\varepsilon'>0$.  
Then for any $U, V \in \mathcal{U}, \mathcal{L}$ or $ \mathcal{T}$ and $(x,t) \in [0, T) \times \R^3$
\begin{align*}
(1+t+ \qv) \varpi (q) \vert \p \phi (t,x) \vert_{UV} & \lesssim \sup_{0 \leq \tau \leq t} \sum_{ \vert I \vert \leq 1} \Vert \vert \varpi (q) Z^I \phi (\tau, \cdot) \vert_{UV} \Vert_{L^\infty} \\
& +  \int_0^t \Big\lbrace \varepsilon' \alpha \Vert \vert \varpi (q) \p \phi (\tau, \cdot) \vert_{UV} \Vert_{L^\infty} + (1+\tau) \Vert \vert \varpi (q) F(\tau, \cdot) \vert_{UV} \Vert_{L^\infty(D_\tau)} \\
& \qquad  \qquad+ \sum_{\vert I \vert \leq 2} (1+\tau)^{-1} \Vert \vert \varpi (q) Z^I \phi (\tau, \cdot) \vert_{UV} \Vert_{L^\infty(D_\tau)} \ \Big\rbrace d\tau ,,
\end{align*}
where $\alpha := \max(1+\gamma', 1/2 - \mu')$. Note here $\vert H \vert_\UU = \sum_{U, V \in \mathcal{U}} \vert U^\mu V^\nu m_{\mu \rho} m_{\nu \sigma} H^{\rho \sigma} \vert$.
\end{corollary}
\begin{remark}
The above estimate is obtained for $\phi_\mn$ and then we contract with any $U, V \in \mathcal{U}$ since $\p_s$ and $\p_q$, defined in Section \ref{section:null-frame}, commute with $\mathcal{U}$. Similarly we could have considered a system $\tbox_g \phi_k = F_k$ and, for the estimate, contracted with arbitrary coefficients $N^k \in \R^m$. 
\end{remark}

We next turn to a generalised version of the Klainerman-Sobolev inequality. Recall from \eqref{eq:intro-weight} the \textit{energy weight} function $w$ 
\begin{equation}
w(q) = \Bigg\lbrace \begin{array}{ll}
1+(1+\qv)^{1+2\gamma} \,, & q>0 \\
1+(1+\qv)^{-2 \mu} \,, & q<0 
\end{array} \,,\label{eq:def-w}
\end{equation}
where $\mu >0$ and $\gamma \in (0, 1/2)$ are fixed constants. This choice of $w$ and $\mu>0$ guarantees $w'(q) >0$, a necessary sign needed to close the energy argument in Theorem \ref{theorem-11.1}.  
This choice of $w$ leads to the following weighted version of the Klainerman-Sobolev inequality. 
\begin{proposition}[Proposition 14.1 from \cite{LR:04}] \label{prop-14.1}
For any function $\phi \in C^\infty_0(\R^3_x)$, at an arbitrary point $(t,x)$ we have
\begin{align*}
\vert \phi(t,x) \vert ( 1+t+\qv) \left( (1+\qv) w(q) \right)^{1/2} & \leq C \sum_{\vert I \vert \leq 3} \Vert w^{1/2} Z^I \phi(t, \cdot) \Vert_{L^2_x} \,.
\end{align*}
\end{proposition}
\begin{remark} \label{remark-ks-ineq-kk}
In the $n=0$ mode reduction of Kaluza-Klein on a $\tee$, the proof of Proposition \ref{prop-14.1} follows the same as in \cite{LR:04} after noting that all the functions considered should be taken independent of the torus coordinates $\{ x^\abar\}$. Thus we may assume that any $\phi$ solving $\tbox_g \phi = F$ must also satisfy $\p_\abar \phi = 0$. Hence any spatial integral becomes
$$ \int_{\Sigma_t} \int_{\tee} \phi \, d^3 x d^m x^\abar = C \int_{\Sigma_t} \phi \, d^3x \,, $$
and so the weighted Klainerman-Sobolev inequality would remain unchanged up to a constant. 
\end{remark}
We now make a bootstrap assumption in order to use Proposition \ref{prop-14.1} to obtain the weak decay estimates. Recall the initial data involved a Schwarzschildean part being subtracted, see \eqref{eq:intro-def-gamma1}. Similarly we must subtract off the Schwarzschildean part from the perturbation. As in \cite{LR:04}, define
\begin{equation} \begin{split} h^0_\mn &:= \chi(r) \chi (r/t) \frac{M}{r} \delta_\mn \,, \\
h^1_\mn &:= h_\mn - h^0_\mn \,, \end{split} \label{eq:def-h0-h1}
\end{equation}
where $\chi(s) \in C^\infty$ is 1 when $s \geq 3/4$ and 0 when $s \leq 1/2$. See also Figure \ref{fig:h0}. 
\begin{remark}
In the case of our zero mode reduction of Kaluza-Klein, we take $h^0_\mn = 0$ whenever at least one of the components $\mu, \nu$ are in the compact directions. 
\end{remark}

\begin{figure}[h] 
\includegraphics[scale=0.7]{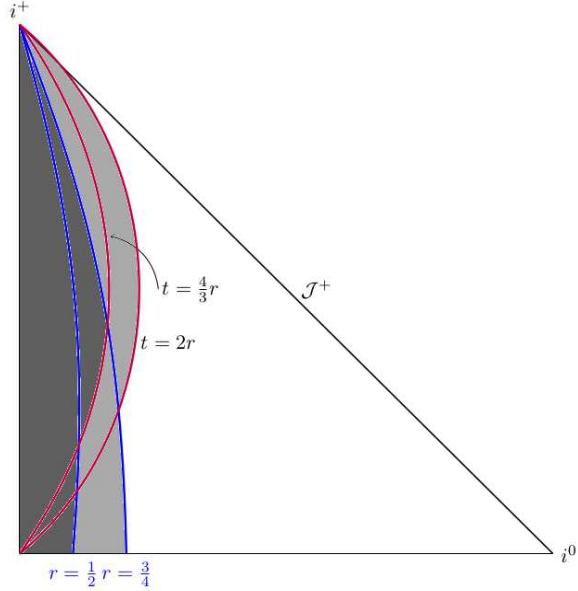}
\caption{Regions where $\chi(r) \chi(r/t)$ is vanishing (dark grey), between 0 and 1 (light grey) and identically 1 (white). In particular, $h^0$ is identically 1 towards $i^0$ and vanishing when $r=0$. }
\label{fig:h0}
\end{figure}

Since $h^0$ is a known quantity, we introduce some notation to separate between the decay for terms involving $h^0$ and the decay for $h^1$. Hence for $i=0,1$ we define
\begin{equation} \begin{split}
W^i & := \left\lbrace U^\mu V^\nu h^i_\mn \,, \psi_k : U, V \in \mathcal{U}  \,, k \in \K \right\rbrace  \\
W_G^i &:=  \left\lbrace T^\mu U^\nu h^i_\mn \,, \psi_k : U \in \mathcal{U} \,, T \in \mathcal{T} \,, k \in \K \right\rbrace \,.
\end{split} \label{eq:def-Wi-Yi} \end{equation}
Recall also the weighted energy from \eqref{eq:intro-energy-sol} for the solution $(h^1_\mn(t), \psi_k(t))$
\begin{equation} \begin{split} \mathcal{E}_N [W^1](t)  = \sup_{0 \leq \tau \leq t} & \sum_{\vert I \vert \leq N} \sum_{Z \in \mathcal{Z}} \int_{\Sigma_\tau} w \left( \vert \p Z^I h^1 (\tau, x ) \vert_\UU^2+  \vert \p Z^I \psi_k(\tau, x ) \vert_\K^2\right) d^3 x \,. \label{eq:def-weighted-energy}
\end{split} \end{equation}
As discussed after Theorem \ref{theorem:intro1}, we make the following bootstrap assumption. Define the time $T \in (0, T_0]$ to be the maximal time such that the following inequality holds
\begin{equation} \begin{split}
\mathcal{E}_N [W^1](t) & \leq 2 C_N \varepsilon (1+t)^\delta \,, 
\end{split} \label{eq:growth-control-h1}
\end{equation}
for $\delta \in (0, 1/4)$ a fixed number such that $ \delta < \gamma$ and with $(m_\mn + h^0_\mn (t)+ h^1_\mn(t), \psi_k(t))$ being a local-in-time solution to \eqref{eq:general-pde} satisfying the wave-coordinate condition 
\begin{equation}
\p_\rho \left( g^{\rho \mu} \sqrt{\vert \det g \vert} \right) = 0  \quad \forall \, t \in [0, T_0] \,. \label{eq:wave-coord}
\end{equation}

Using the weighted Klainerman-Sobolev inequality and the bootstrap assumption one can now derive the following `weak decay' estimates. 
\begin{corollary}[Modified Corollary 9.4 from \cite{LR:04}] \label{corol-9.4}
Assume \eqref{eq:growth-control-h1} holds and $W^i$ is defined as in \eqref{eq:def-Wi-Yi}, then for $i=0,1$ and $\I \leq N-2$ we have
\begin{align}
\vert Z^I W^i(t,x) \vert & \leq \Bigg\lbrace \begin{array}{ll}
C \varepsilon (1+t+\qv)^{-1+\delta} (1+\qv)^{-\delta'} \,, & q >0 \\
C \varepsilon (1+t+\qv)^{-1+\delta} (1+\qv)^{1/2} \,, & q <0 
\end{array} \,, \label{eq:corol-9.4-Zh} \\
\vert \p Z^I W^i(t,x) \vert & \leq \Bigg\lbrace \begin{array}{ll}
C \varepsilon (1+t+\qv)^{-1+\delta} (1+\qv)^{-1-\delta'} \,, & q >0 \\
C \varepsilon (1+t+\qv)^{-1+\delta} (1+\qv)^{-1/2} \,, & q <0 
\end{array} \,, \label{eq:corol-9.4-pZh}
\end{align}
where $\delta'=\delta$ if $i=0$ and $\delta'=\gamma > \delta$ if $i=1$.
Furthermore if $\I \leq N-3$ then
\begin{align}
\vert \pgood Z^I W^i(t,x) \vert & \leq \Bigg\lbrace \begin{array}{ll}
C \varepsilon (1+t+\qv)^{-2+\delta} (1+\qv)^{-\delta'} \,, & q >0 \\
C \varepsilon (1+t+\qv)^{-2+\delta} (1+\qv)^{1/2} \,, & q <0 
\end{array} \,. \label{eq:corol-9.4-pgoodZh}
\end{align}
\end{corollary}  
\begin{proof}
The proof of \eqref{eq:corol-9.4-pZh} follows from the weighted Klainerman-Sobolev estimate of Proposition \ref{prop-14.1} and the energy assumption  \eqref{eq:growth-control-h1}. We obtain \eqref{eq:corol-9.4-Zh} for $r>t$  by integrating \eqref{eq:corol-9.4-pZh} from the hypersurface $t=0$ along lines with $t+r$ and $w = x/\vert x \vert$ fixed. One also needs the initial condition from the assumptions of our main Theorem \ref{theorem:intro1}, namely \eqref{eq:intro-decay-id} implies
$$ \liminf_{\vert x \vert \to \infty} \left( \vert h^1(0,x) \vert + \vert \psi_k(0,x) \vert \right) = 0 \,. $$
The final estimates for $\vert \pgood Z^I W^i(t,x) \vert$ come from recalling that for any $\phi \in C^\infty_c$ function
$$\vert \pgood \phi \vert \lesssim (1+t+\qv)^{-1}  \sum_{\I = 1} \vert Z^I \phi \vert \,.$$
\end{proof}
These weak decay estimates immediately give the following decay on the inhomogeneities. 
\begin{lemma}[Modified Lemma 10.5 from \cite{LR:04}] \label{lemma-10.5}
Assume \eqref{eq:growth-control-h1} and \eqref{eq:wave-coord} hold and let $F_k$ be as given in \eqref{eq:general-pde-F}. Then
\begin{align*}
\vert F \vert_\TU + \vert F \vert_\K & \lesssim  \varepsilon t^{-3/2 + \delta} \vert \p W \vert \,, \\
\vert F \vert_\UU &\lesssim  \varepsilon  t^{-3/2 + \delta} \vert \p W \vert +  \vert \p W \vert_G^2\,.
\end{align*}
\end{lemma}
\begin{proof}
This follows from Proposition \ref{prop-9.8} and the weak decay estimates of Corollary \ref{corol-9.4}.
\end{proof}

We end this section with a stand-alone Lemma which involves estimating the second-order derivatives landing on the Schwarzschildean part $h^0_\mn$. Since the form of $h^0_\mn$ has been chosen in \eqref{eq:def-h0-h1}, this result comes from commuting through copies of $Z \in \mathcal{Z}$ and applying the results from Corollary \ref{corol-9.4}.  
\begin{lemma}[Lemma 9.9 from \cite{LR:04}] \label{lemma-9.9}
Let $F^0_\mn := \tbox_g h^0_\mn$ where $h^0$ is as defined in \eqref{eq:def-h0-h1}. Then for $\I \leq N-2$
\begin{align*}
\vert Z^I F^0 \vert_{\mathcal{U} \mathcal{U}} & \leq \Bigg\lbrace \begin{array}{ll}
C \varepsilon^2 (1+t+\qv)^{-4+\delta} (1+\qv)^{-\delta} \,, & q >0 \,,  \\
C \varepsilon (1+t+\qv)^{-3} \,, & q <0 \,,
\end{array} 
\end{align*}
and for $\I \leq N$
\begin{align*}
\vert Z^I F^0 \vert_\UU & \lesssim \Bigg\lbrace \begin{array}{ll}
 C \varepsilon^2 (1+t+\qv)^{-4+\delta} (1+\qv)^{-\delta} \,, & q >0  \\
C \varepsilon (1+t+\qv)^{-3} \,, & q <0 
\end{array} \\
& \qquad + \frac{C \varepsilon}{(1+t+\qv)^{3}} \sum_{\J \leq \I} \vert Z^J h^1 \vert_\UU \,.
\end{align*}
\end{lemma}

\section{Decay Estimates -- Part II} \label{section:decay-2}
In this section we will prove improved decay estimates for $\p Z^I W$, valid only for a smaller number of $\mathcal{Z}$ vector fields, namely for $\I \leq N/2 + 2$. We first obtain estimates for $\vert \p W \vert$ and $\vert \p W \vert_G$ given in Proposition \ref{prop-10.1a}, followed by estimates of $\p Z^I h$ for $\I \leq 1$ in Proposition \ref{prop-10.1b}. 

\begin{lemma}[Modified Lemma 10.6 from \cite{LR:04}] \label{lemma-10.6} 
Let $(h_\mn, \psi_k)$ be a local-in-time solution of \eqref{eq:general-pde} such that \eqref{eq:growth-control-h1} and \eqref{eq:wave-coord} hold. Let $F_\mn$ and $F_k$ be as defined in \eqref{eq:general-pde-F}. Then
\begin{align*}
(1+t) \Vert \p W_G (t, \cdot) \Vert_{L^\infty} & \leq C \varepsilon + C \varepsilon \int_0^t (1+\tau)^{\delta-1/2} \Vert \p W(\tau, \cdot)  \Vert_{L^\infty} d \tau \,, \\
(1+t) \Vert \p W (t, \cdot) \Vert_{L^\infty} & \leq C \varepsilon + C \int_0^t \Big\lbrace \varepsilon (1+\tau)^{\delta - 1/2} \Vert  \p W ( \tau, \cdot)  \Vert_{L^\infty} + (1+\tau) \Vert \p W_G (\tau, \cdot) \Vert^2_{L^\infty} \Big\rbrace d \tau \,.
\end{align*}
\end{lemma}
\begin{proof}
Apply Corollary \ref{corol-7.2} with $\varpi(q) =1$, ie, $\alpha = 0$, and estimate $\vert F \vert_\TU, \vert F \vert_\K$ and $\vert F \vert_\UU$ with Lemma \ref{lemma-10.5} and the weak decay estimates of Corollary \ref{corol-9.4}.
\end{proof}

The following algebraic lemma is needed for Proposition \ref{prop-10.1a}. 
\begin{lemma}[Lemma 10.7 from \cite{LR:04}] \label{lemma-10.7}
If $b(t) \geq 0$ and $c(t) \geq 0$ satisfy
\begin{align*}
b(t) & \leq C \varepsilon \left( \int_0^t (1+s)^{-1-a} c(s) ds + 1 \right) \,, \\
c(t) & \leq C \varepsilon \left( \int_0^t (1+s)^{-1-a} c(s) ds + 1 \right) + C \int_0^t (1+s)^{-1} b^2(s) ds \,,
\end{align*}
for some positive constants such that $a \geq C^2 \varepsilon$ and $a \geq 4 C \varepsilon / (1-2C\varepsilon)$, then
$$ b(t) \leq 2 C \varepsilon \,, \quad \text{ and } \quad c(t) \leq 2C \varepsilon (1+ a \ln (1+t)) \,.$$
\end{lemma}

\begin{proposition}[Modified Proposition 10.1 from \cite{LR:04}] \label{prop-10.1a}
Let $(h_\mn, \psi_k)$ be a local-in-time solution of \eqref{eq:general-pde} such that \eqref{eq:growth-control-h1} and \eqref{eq:wave-coord} hold. Then
\begin{align*}
\vert \p W \vert_G & \leq C \varepsilon (1+t+\qv)^{-1} \,, \\
\vert \p W \vert & \leq C \varepsilon t^{-1} \ln t \,.
\end{align*}
\end{proposition}
\begin{proof}
Take $a = 1/2 - \delta$, $b(t) := (1+t) \Vert \p W_G (\tau, \cdot)  \Vert_{L^\infty}$ and $c(t) := (1+t) \Vert \p W (\tau, \cdot) \Vert _{L^\infty}$ in Lemma \ref{lemma-10.7}.
\end{proof}
In the original paper by Lindblad and Rodnianski, it was shown that all but one component of $\p h_\mn$ obeys a $(1+t+\qv)^{-1}$ decay rate. The `worst' component $\p h_{\lbar \lbar}$ has the slower $t^{-1} \ln t$ decay rate. This was the decay rate originally suspected for all the components. Our choice of additional non-linearity, as discussed in Remark \ref{remark:loizelet}, has meant that the new variables $\psi_k$ pick up the better of these two decay rates.  

We now turn to estimates of the best components of $\p Z^I h$ for $\I \leq 1$, given in Proposition \ref{prop-10.1b}. These `best' components have much better decay rates than the full set given before in Proposition \ref{prop-10.1a}. Note however that since we are using control obtained from the wave-coordinate condition condition, we only obtain estimates on $h_\mn$ and not $\psi_k$. Thus Lemma \ref{lemma-10.4} and Proposition \ref{prop-10.1b} are completely unchanged from \cite{LR:04}. 
 
 \begin{lemma}[Lemma 10.4 from \cite{LR:04}] \label{lemma-10.4}
Let $(h_\mn, \psi_k)$ be a local-in-time solution of \eqref{eq:general-pde} such that \eqref{eq:growth-control-h1} and \eqref{eq:wave-coord} hold. Then
\begin{align*}
\sum_{\I \leq k} \vert \p Z^I h \vert_\LL + \sum_{\I \leq k-1} \vert \p Z^I h \vert_\LT & \lesssim \sum_{\I \leq k-2} \vert \p Z^I h \vert_\UU \\ & \qquad + \Bigg\lbrace \begin{array}{ll}
\varepsilon (1+t+\qv)^{-2-2 \delta} (1+\qv)^{-2 \delta} \,, & q>0 \\
\varepsilon (1+t+\qv)^{-2-2 \delta} (1+\qv)^{-1/2-\delta} \,, & q<0 
\end{array} \,,
\end{align*}
\begin{align*}
\sum_{\I \leq k} \vert Z^I h \vert_\LL + \sum_{\I \leq k-1} \vert Z^I h \vert_\LT & \lesssim \sum_{\I \leq k-2} \int_{s,\omega = const} \vert Z^I h \vert_\UU \\ & \qquad + \Bigg\lbrace \begin{array}{ll}
\epsilon (1+t+\qv)^{-1} & \,, q>0 \\
\epsilon (1+t+\qv)^{-1} (1+\qv)^{1/2+\delta} & \,, q<0 
\end{array}\,,
\end{align*}
where here the sums over $k-2$ (resp. $k\leq 1$) are absent if $k \leq 1$ (resp. $k=0$).
\end{lemma}

\begin{proposition}[Proposition 10.1 from \cite{LR:04}] \label{prop-10.1b}
Let $(h_\mn, \psi_k)$ be a local-in-time solution of \eqref{eq:general-pde} such that \eqref{eq:growth-control-h1} and \eqref{eq:wave-coord} hold. Then
\begin{align*}
\vert \p h \vert_{\mathcal{L} \mathcal{T}}  + \vert \p Z h \vert_\LL & \leq \Bigg\lbrace \begin{array}{ll}
\varepsilon (1+t+\qv)^{-2 + 2 \delta}(1+q)^{-\delta} & \,, q>0 \\
\varepsilon (1+t+\qv)^{-2 + 2 \delta}(1+q)^{1/2} & \,, q<0 
\end{array} \,, \\
\vert  h \vert_{\mathcal{L} \mathcal{T}}  + \vert Z h \vert_\LL & \leq \Bigg\lbrace \begin{array}{ll}
\varepsilon (1+t+\qv)^{-1} & \,, q>0 \\
\varepsilon (1+t+\qv)^{-1}(1+q)^{1/2+\delta} & \,, q<0 
\end{array} \,.
\end{align*}
\end{proposition}

We now continue to the crux of this section, and derive `stronger' estimates of all components of $\p Z^I W$ for the restricted range $\I \leq N/2 + 2$. The proof relies on using Corollary \ref{corol-7.2} again, this time with a non-trivial weight $\varpi$. Recall that in Corollary \ref{corol-7.2} there was a spacetime integral of $ \vert \varpi F \vert$ where $F$ is the particular inhomogeneity being considered. 
The inhomogeneity we consider comes from commuting $Z^I$ through the PDE \eqref{eq:general-pde}. This leads to the system
\begin{equation} \begin{split}
\tbox_g Z^I h^1_\mn &= D^I h^1_\mn + \hat{Z}^I F_\mn - \hat{Z}^I F^0_\mn  \,, \\
\tbox_g Z^I \psi_k &= D^I \psi_k + \hat{Z}^I F_k \,.
\end{split} \label{eq:prop-10.2-pde-system}
\end{equation}
As in Lindblad and Rodnianski, we have introduced the shifted vector field $\hat{Z} := Z + c_Z$ where $c_Z$ is a constant defined by $\p_\mu Z_\nu + \p_\nu Z_\mu = c_Z m_\mn$. In particular, $c_Z = 0$ except when $Z = S$ in which case $c_S = 2$. The commutator-like term $D^I$ is defined by $$ D^I := \tbox_g Z - \hat{Z} \tbox_g \,. $$
The estimates of $D^I h^1_\mn$ and $D^1 \psi_k$ can be found in Proposition \ref{prop-5.3}. On the right hand side of \eqref{eq:prop-10.2-pde-system} are the terms $\hat{Z}^I F_\mn$ and $\hat{Z}^I F_k$ which are estimated in the following lemma. 

\begin{lemma}[Modified Lemma 10.8 from \cite{LR:04}] \label{lemma-10.8} Let $(h_\mn, \psi_k)$ be a local-in-time solution of \eqref{eq:general-pde} such that \eqref{eq:growth-control-h1} and \eqref{eq:wave-coord} hold. Let $F_\mn$ and $F_k$ be as given in \eqref{eq:general-pde-F}. Then for $\I \leq N-2$ we have
\begin{align*}
\vert Z^I F \vert_\K + \vert Z^I F \vert_\UU & \lesssim  \varepsilon (1+t)^{-1}\sum_{\vert J \vert \leq \vert I \vert} \vert \p Z^J W \vert +  \sum_{\vert J \vert + \vert K \vert \leq \vert I \vert, \vert K \vert < \vert I \vert} \vert \p Z^J W \vert \vert \p Z^K W \vert  \,,
\end{align*}
where it is understood that the term with $|K| < \I$ vanishes if $\I = 0$. 
\end{lemma}

\begin{proof}
The result follows from Proposition \ref{prop-9.8}, Corollary \ref{corol-9.4} and the first estimate in Proposition \ref{prop-10.1a}.
\end{proof}

\begin{proposition}[Modified Proposition 10.2 from \cite{LR:04}] \label{prop-10.2} Let $(h_\mn, \psi_k)$ be a solution of the generalised PDE \eqref{eq:general-pde}. 
Assume \eqref{eq:growth-control-h1} holds, $W^i$ is defined as in \eqref{eq:def-Wi-Yi} and $F_\mn$ and $F_k$ are defined as in \eqref{eq:general-pde-F}. 
Let $\gamma' < \gamma - \delta$ and $\mu' > \delta>0$ be fixed. Then there exist constants $M_k, C_k$ and $\varepsilon$ depending on $\gamma', \mu'$ and $\delta$ such that for all $\I = k \leq N/2 + 2$
\begin{align*}
\vert \p Z^I W^i \vert& \leq \Bigg\lbrace \begin{array}{ll} C_k \varepsilon (1+t+\qv)^{-1 + M_k \varepsilon}(1+\qv)^{-1-\delta'}  \,, & q >0 \\
C_k \varepsilon (1+t+\qv)^{-1 + M_k \varepsilon}(1+\qv)^{-1/2-\mu'}  \,, & q <0 
\end{array} \,, \\
\vert Z^I W^i \vert & \leq \Bigg\lbrace \begin{array}{ll} C_k \varepsilon (1+t+\qv)^{-1 + M_k \varepsilon}(1+\qv)^{-\delta'}   \,, & q >0 \\
C_k \varepsilon (1+t+\qv)^{-1 + M_k \varepsilon}(1+\qv)^{1/2-\mu'} \,, & q<0 
\end{array}  \,,
\end{align*}
where $\delta'=\gamma'$ if $i=1$ or $\delta'=M_k \varepsilon$ if $i=0$. 
\end{proposition}
\begin{proof}
The proof follows by induction, as in \cite{LR:04}. The base case $k=0$ follows in a simpler but similar way to the main case. One assumes the step for $\I \leq k$ and then considers the case of $\I = k+1$. 
From \eqref{eq:prop-10.2-pde-system} the following holds
\begin{align*}
\vert \tbox_g Z^I h^1\vert_\UU
& \leq \vert D^I h^1 \vert_\UU + \vert Z^I F \vert_\UU + \vert Z^I F^0 \vert_\UU \,, \\
\vert \tbox_g Z^I \psi \vert_\K & \leq \vert D^I \psi \vert_\K + \vert Z^I F \vert _\K \,.
\end{align*}
From here we use Proposition \ref{prop-9.8} and Lemma \ref{lemma-10.8} to estimate $\vert Z^I F \vert_\UU$ and $\vert Z^I F \vert_\K$. Lemma \ref{lemma-9.9} gives an estimate for $\vert Z^I F^0\vert_\UU$. Lastly Proposition \ref{prop-5.3} gives an estimate on $\vert D^I h^1\vert_\UU$ and $\vert D^I \psi \vert_\K$. 
One then collects together terms according to whether $\vK=\I$ or $\vK < \I$. When $\vK=\I$ we require strong estimates on the terms with a low number of $Z$'s acting on $H$, ie, of the form
$$ \sum_{\J \leq 1} \left( \vert Z^J H \vert_\LL + \vert H \vert_\LT \right) \,. $$
Since there are no $\psi_k$ terms appearing here, we may use Proposition \ref{prop-10.1b} which gives the necessary decay. 
When $\vK < \I$ one must use the induction hypothesis. 

Putting this altogether gives an estimate on $\vert \tbox_g Z^I W^1 \vert$. Corollary \ref{corol-9.4} implies an estimate on $\varpi (q) \vert Z^I W^1(t,x) \vert$. 
Then inserting all this in Corollary \ref{corol-7.2} yields an integral inequality and the result then follows by Gronwall's inequality.  
\end{proof}

\section{Energy Estimates} \label{section:energy}
In this final section we will obtain an integrated energy inequality for $\mathcal{E}_N[W^1](t)$ using the following Proposition from \cite{LR:04}. Using this we will be able to apply a Gronwall type argument to deduce the improved inequality \eqref{eq:intro-aim} for the energy. 
\begin{proposition}[Proposition 6.2 from \cite{LR:04}] \label{prop-6.2}
Let $\phi$ be a solution to $\tbox_g \phi = F$ such that for $H^\mn := g^\mn - m^\mn$ we have
\begin{equation}
\begin{split}
(1+\qv )^{-1} \vert H \vert_\LL + \vert \p H \vert_\LL + \vert \pgood H \vert_\UU &\leq C \varepsilon' (1+t)^{-1} \,, \\
(1+\qv )^{-1} \vert H \vert_\UU + \vert \p H \vert_\UU &\leq C \varepsilon' (1+t)^{-1/2} (1+\qv)^{1/2} (1+q_-)^{-\mu} \,.
\end{split} \label{eq:prop-6.2-decay-assump}
\end{equation}
Then for any $0 < \gamma \leq 1$ and $0 < \varepsilon' \leq \gamma/C$ we have
\begin{equation} \begin{split}
 \int_{\Sigma_t} \vert \p \phi \vert^2 w d^3 x + \int_0^t  \int_{\Sigma_\tau} \vert \pgood \phi \vert^2 w' d^3 x d \tau & \leq  8 \int_{\Sigma_0} \vert \p \phi \vert^2 w d^3 x \\ & +  16 \int_0^t  \int_{\Sigma_\tau} \varepsilon \left( \frac{\vert \p \phi \vert^2}{1+t} w + \vert F \vert \vert \p \phi \vert w \right) d^3 x d \tau \,. \end{split}
\end{equation}
\end{proposition}
Using Young's inequality and the above Proposition \ref{prop-6.2} on the PDE \eqref{eq:prop-10.2-pde-system} yields
\begin{equation} \begin{split}
& \int_{\Sigma_t }  \left(  \vert \p Z^I h^1 \vert^2_\UU + \vert  \p Z^I \psi \vert^2_\K \right) w d^3x  + \int_0^t \int_{\Sigma_\tau} \left( \vert \pgood Z^I h^1 \vert^2_\UU + \vert \pgood Z^I \psi \vert^2_\K \right) w' d^3 x d \tau \\
& \qquad \leq 8  \int_{\Sigma_0} \left( \vert \p Z^I h^1 \vert^2_\UU + \vert \p Z^I \psi \vert^2_\K \right) w dx + 16 \int_0^t \int_{\Sigma_\tau}  \frac{C \varepsilon}{1+t}  \left( \vert \p Z^I h^1 \vert^2_\UU + \vert \p Z^I \psi \vert^2_\K \right) w  d^3 x d \tau\\
& \qquad + \int_0^t \int_{\Sigma_\tau} \varepsilon^{-1} \left( \vert \hat{Z}^I F \vert^2_\UU + \vert \hat{Z}^I F \vert^2_\K + \vert D^I h^1 \vert^2_\UU + \vert D^I \psi \vert^2_\K  \right) (1+t)w d^3 x d \tau \\
& \qquad + \int_0^t \int_{\Sigma_\tau} \vert Z^I F^0 \vert_\UU \vert \p Z^I h^1 \vert_\UU  w   d^3 x d \tau \,.
\label{eq:prop7.1-energy-W} \end{split} \end{equation}
The terms not `roughly' in the form of the energy $\mathcal{E}_N[W^1](t)$, defined in \eqref{eq:def-weighted-energy}, are 
\begin{equation}
 \varepsilon^{-1} ( \vert \hat{Z}^I F \vert^2_\UU + \vert \hat{Z}^I F \vert^2_\K + \vert D^I h^1 \vert^2_\UU + \vert D^I \psi \vert^2_\K ) (1+t)w + \vert Z^I F^0 \vert_\UU \vert \p Z^I h^1 \vert_\UU  w \label{eq:thm-11.1-terms-to-control}
\end{equation}
and so these are estimated in the following Lemmas \ref{lemma-11.2} to \ref{lemma-11.5}. First we state the main Theorem \ref{theorem-11.1} proving \eqref{eq:intro-aim} as the assumptions will be the same for Lemmas \ref{lemma-11.2} to \ref{lemma-11.5}, but leave the proof until after the Lemmas. 

\begin{theorem}[Modified Theorem 11.1 from \cite{LR:04}] \label{theorem-11.1}
Let $(h_\mn, \psi_k)$ be a local in time solution to \eqref{eq:general-pde} satisfying the wave gauge condition \eqref{eq:wave-gauge} on some maximal interval $[0, T)$. Suppose for fixed $\mu' \in (0, 1/2) $ and $\gamma \in (0,1/2)$ that for all $t \in [0,T]$ and all $\I \leq N/2 + 2$ we have
\begin{subequations}
\label{eq:thm-11.1-asmpt}
 \begin{align}
\vert \p W \vert_G + (1+\qv)^{-1} \vert H \vert_{\mathcal{T} \mathcal{L}} + &(1+\qv)^{-1} \vert Z H \vert_\LL  \leq C \varepsilon (1+t)^{-1} \label{eq:thm-11.1-assumpt1} \,, \\
\vert \p Z^I W \vert + \frac{\vert Z^I W \vert}{1+\qv} + \frac{1+t+\qv}{1+\qv} \vert \pgood Z^I W \vert & \leq \Big\lbrace \begin{array}{ll}
C\varepsilon (1+t+\qv)^{-1+C\varepsilon} (1+\qv)^{-1-C\varepsilon} \,, \\
C\varepsilon (1+t+\qv)^{-1+C\varepsilon} (1+\qv)^{-1/2+\mu'}
\end{array}  \label{eq:thm-11.1-assumpt2} \\
E_N[W^1](0) + M^2 & \leq \varepsilon^2 \,. \label{eq:thm-11.1-assumpt3} 
 \end{align}
\end{subequations}
For any constant $\mu$ satisfying $0<\mu<1/2-\mu'$ there exist constants $C_N$ and $C$, independent of $T$, such that for $\varepsilon>0$ sufficiently small the following holds for all $t \in [0, T]$
\begin{equation} \mathcal{E}_N[W^1](t) \leq C_N \varepsilon ^2 (1+t)^{C\varepsilon} \,. \end{equation}
\end{theorem}

We begin by estimating the first pair of terms of \eqref{eq:thm-11.1-terms-to-control} in Lemmas \ref{lemma-11.2} and \ref{lemma-11.3}. 
\begin{lemma}[Modified Lemma 11.2 from \cite{LR:04}] \label{lemma-11.2}
Assuming the conditions of Theorem \ref{theorem-11.1} then
\begin{align*}
\vert Z^I F \vert_\UU & \lesssim \sum_{\vert J \vert \leq \vert I \vert} \Big( \frac{\varepsilon \vert \p Z^J W^1 \vert}{1+t} + \frac{\varepsilon (1+\qv)^{\mu'-1/2}}{(1+t+\qv)^{1-C\varepsilon}} \vert \pgood Z^J W^1 \vert + \frac{\varepsilon^2}{1+t+\qv} \frac{\vert Z^J W^1 \vert}{1+\qv} \Big)  \\
&+ \sum_{\vert J \vert \leq \vert I \vert -1} \frac{\varepsilon}{(1+t)^{1-C\varepsilon}} \vert \p Z^J h^1 \vert_\UU + \frac{\varepsilon^2}{(1+t+\qv)^4} \,, \\
\vert Z^I F \vert_\K & \lesssim \sum_{\vert J \vert \leq \vert I \vert} \Big( \frac{\varepsilon \vert \p Z^J W^1 \vert}{1+t} + \frac{\varepsilon (1+\qv)^{\mu'-1/2}}{(1+t+\qv)^{1-C\varepsilon}} \vert \pgood Z^J W^1 \vert + \frac{\varepsilon^2}{1+t+\qv} \frac{\vert Z^J W^1 \vert}{1+\qv} \Big)  \\
&+  \frac{\varepsilon^2}{(1+t+\qv)^4} \,.
\end{align*}
\end{lemma}
\begin{proof}
The proof follows from Proposition \ref{prop-9.8} and by estimates coming from the assumptions \eqref{eq:thm-11.1-assumpt1}, \eqref{eq:thm-11.1-assumpt2} and from calculating $\p Z^I h^0, Z^I h^0$. 
\end{proof}

\begin{lemma}[Modified Lemma 11.3 from \cite{LR:04}] \label{lemma-11.3}
Assuming the conditions of Theorem \ref{theorem-11.1} then
\begin{align*}
\varepsilon^{-1} \int_0^t \int_{\Sigma_\tau} & \left( \vert Z^I F \vert^2_\UU + \vert Z^I F \vert^2_\K \right) (1+\tau) w d \tau d^3 x \\
& \lesssim \sum_{\J \leq \I} \int_0^t \int_{\Sigma_\tau} \varepsilon \Big( \frac{\vert \p Z^J W^1 \vert^2 }{1+\tau}w + \vert \pgood Z^J W^1 \vert^2 w' \Big) d \tau d^3 x\\
& + \sum_{\J \leq \I  -1 } \int_0^t \int_{\Sigma_\tau} \varepsilon \frac{\vert \p Z^J h^1 \vert^2{}_\UU}{(1+\tau)^{1-2C \varepsilon}} w d \tau d^3 x + \varepsilon^3 \,.
\end{align*}
\end{lemma}
\begin{proof}
The estimate from Lemma \ref{lemma-11.2} and then applying the variant of Hardy's inequality given in Corollary \ref{corol-13.3} on the $\vert Z^J W^1 \vert$ term. 
\end{proof}

Next we estimate the second pair of terms from \eqref{eq:thm-11.1-terms-to-control}. 
\begin{lemma}[Modified Lemma 11.5 from \cite{LR:04}] \label{lemma-11.5}
Assuming the conditions of Theorem \ref{theorem-11.1} then
\begin{align*}
 \varepsilon^{-1} \int_0^t \int_{\Sigma_\tau} & \left( \vert D^I h^1 \vert^2_\UU + \vert D^I \psi \vert^2_\K \right) (1+\tau) w dx dt \\
 &\lesssim \varepsilon \sum_{\J \leq \I} \int_0^t \int_{\Sigma_\tau} \Big( \frac{w}{1+t} \vert \p Z^J W^1 \vert^2 + \vert \pgood Z^J W^1 \vert^2 w' \Big) dx dt \\ & \quad + \varepsilon \sum_{\J \leq \I -1} \int_0^t \int_{\Sigma_\tau} \frac{w}{(1+t)^{1-2C\varepsilon}} \vert \p Z^J W^1 \vert^2 d \tau d^3 x+ \varepsilon ^3 \,.
\end{align*}
\end{lemma}
\begin{proof}
The proof follows as in \cite{LR:04} by using the decay estimates for $\vert \p Z^I h \vert_\UU$ and $\vert \p Z^I \psi \vert_\K$ given in \eqref{eq:thm-11.1-assumpt2}. The stronger estimates required on $\vert Z^J H \vert _\LL$ and $\vert H \vert_\LT$, are unchanged and given in \eqref{eq:thm-11.1-assumpt1}. For the full technical details we refer to \cite{LR:04}. 
\end{proof}

For the last term in \eqref{eq:thm-11.1-terms-to-control}, recall that $h^0_\mn$, and thus $Z^I \tbox_g h^0_\mn$, is determined by calculating directly from the definition in \eqref{eq:def-h0-h1}. Thus the next Lemma follows identically from \cite{LR:04}, with the proof using Lemma \ref{lemma-9.9} and Corollary \ref{corol-13.3} on the $\vert Z^I h^1 \vert$ term. 
\begin{lemma}[Lemma 11.4 from \cite{LR:04}] \label{lemma-11.4}
Assuming the conditions of Theorem \ref{theorem-11.1} then
\begin{align*}
\int_0^t & \int_{\Sigma_\tau} \vert Z^I F^0 \vert_\UU  \vert \p Z^I h^1 \vert_\UU w d \tau d^3 x \\ &
\lesssim \varepsilon \sum_{\vert J \vert \leq \vert I \vert}  \Big( \int_0^t \int_{\Sigma_\tau} \frac{\vert \p Z^J h^1 \vert^2_\UU}{(1+\tau)^2} w d \tau d^3 x+ \int_0^t \left( \int_{\Sigma_\tau} \vert \p Z^J h^1 \vert^2_\UU w d^3x \right)^{1/2} \frac{d \tau}{(1+\tau)^{3/2}} \Big) \,.
\end{align*}
\end{lemma}
We can now complete the proof of the main theorem. 
\begin{proof}[Proof of Theorem \ref{theorem-11.1}]
First we discuss some features of the volume form. For some function $\phi$, the inequality
\begin{align*}
 \left\vert \int_{\Sigma_t} \phi \vert \det g \vert^{1/2} d^3 x -  \int_{\Sigma_t} \phi \vert \det m \vert^{1/2} d^3 x \right\vert \leq \frac{1}{4} \left\vert \int_{\Sigma_t} \phi \, d^3 x \right\vert \,,
\end{align*}
holds since by Corollary \ref{corol-9.4} $\vert H \vert_\UU \leq C \varepsilon $ and 
$$\vert \det  g \vert^{1/2} = \vert \det  m \vert^{1/2} - \frac{1}{2}\tr_{m} H + \mathcal{O}(H^2)=1- \frac{1}{2}\tr_{m} H + \mathcal{O}(H^2) \,.$$
This means, up to some constant, we can use the background Minkowski volume form but still control an energy $\mathcal{E}_N[W^1](t)$ with volume form actually coming from the perturbed metric. 
Then precisely as in \cite{LR:04}, Lemmas \ref{lemma-11.3}, \ref{lemma-11.4}, \ref{lemma-11.5} and the integrated energy inequality \eqref{eq:prop7.1-energy-W} imply
\begin{equation} \begin{split}
& \int_{\Sigma_t }  \left(  \vert \p Z^I h^1 \vert^2_\UU + \vert  \p Z^I \psi \vert^2_\K \right) w d^3x  + \int_0^t \int_{\Sigma_\tau} \left( \vert \pgood Z^I h^1 \vert^2_\UU + \vert \pgood Z^I \psi \vert^2_\K \right) w' d^3 x d \tau \\
&  \leq 8  \int_{\Sigma_0} \left( \vert \p Z^I h^1 \vert^2_\UU + \vert \p Z^I \psi \vert^2_\K \right) w dx + C \varepsilon \sum_{\J \leq \I} \int_0^t \int_{\Sigma_\tau} \frac{\vert \p Z^J W^1 \vert^2}{1+ \tau} w d^3x d\tau \\
& \qquad + C \varepsilon \sum_{\J \leq \I} \int_0^t \int_{\Sigma_\tau} \vert \pgood Z^J W^1 \vert^2 w' d^3x d\tau  + C \varepsilon \sum_{\J \leq \I -1} \int_0^t \int_{\Sigma_\tau} \frac{\vert \p Z^J W^1 \vert^2}{(1+\tau)^{1-2C\varepsilon}} w d^3 x d \tau \\
& \qquad + C_N \varepsilon \sum_{\J \leq \I} \int_0^t \left( \int_{\Sigma_\tau} \vert \p Z^J h^1 \vert^2_\UU w d^3 x \right)^{1/2} \frac{d \tau}{(1+\tau)^{3/2}} + C \varepsilon^3 \,.
\end{split} \end{equation}
Similar to \cite{LR:04}, define a quantity for the `good' derivative
\begin{align*} 
S_m(t) &:= \sum_{\vert I \vert \leq m} \int_0^t \int_{\Sigma_\tau} \left( \vert \pgood Z^I h^1 \vert_\UU^2 +  \vert \pgood Z^I \psi \vert^2_\K \right) w' d^3x \,.
\end{align*}
Note that apart from the derivative, this involves a spacetime integral compared to \eqref{eq:def-weighted-energy} which only has a spatial integral. Then for $m \leq N$ we obtain the inequality
\begin{align*}
\mathcal{E}_m(t) + S_m(t) & \leq 8 \mathcal{E}_m(0) + C \varepsilon S_m(t) + C \varepsilon^3 \\
 & +C_N \varepsilon \int_0^t\frac{\mathcal{E}_m(\tau)^{1/2}}{(1+\tau)^{3/2}} d \tau + C \varepsilon \int_0^t \frac{\mathcal{E}_m(\tau)}{1+\tau} d \tau + C \varepsilon \int_0^t \frac{\mathcal{E}_{m-1} (\tau)}{(1+\tau)^{1-C\varepsilon}}  d \tau \,.
\end{align*}
One can now follow the argument in \cite{LR:04} and use Gronwall's inequality and induction to deduce
$$ \mathcal{E}_m (t) \leq C_N \varepsilon^2 (1+t)^{2C \varepsilon} \,,$$
for all $m \leq N$. Since the constants are independent of time we have
$$ \mathcal{E}_N (t) \leq C_N \varepsilon^2 (1+t)^{C\varepsilon} \,, $$
for all $0 \leq t \leq T$. 
\end{proof}

\begin{remark}  
As in Remark \ref{remark-ks-ineq-kk}, the Kaluza-Klein example will have all variables independent of the torus coordinates. Thus any integrals in the above proof of Theorem \ref{theorem-11.1} will pick up an additional, non-zero, constant from the volume of the torus. 
\end{remark}

\begin{theorem}[Modified Theorem 9.1 from \cite{LR:04}] \label{thm-9.1}Let $( \Sigma_0,\gamma_{ij}, K_{ij}, f_k, g_k)$ be smooth initial data for \eqref{eq:general-pde},  asymptotically flat in the sense of \eqref{eq:intro-decay-id}, with $\Sigma_0$ diffeomorphic to $\R^3$. Then there exists a constant $\varepsilon_0>0$ such that for all $\varepsilon \leq \varepsilon_0$ and initial data such that 
$$ E_N (0)^{1/2} + M \leq \varepsilon \,, $$
the solution $(m_\mn + h_\mn(t), \psi_k(t))$ to the system \eqref{eq:general-pde} can be extended to a global-in-time smooth solution  in which the energy satisfies
\begin{equation} \mathcal{E}_N[W^1](t) \leq C_N \varepsilon (1+t)^{C \varepsilon} \,.  \end{equation}
where $C_N$ and $C$ depend only on $N$.
\end{theorem}
\begin{proof}
It holds that
$$ \mathcal{E}_N[W^1](0) \lesssim E_N[W^1](0) + M^2 \lesssim \varepsilon^2 \,. $$
Hence there exists a maximal $T>0$ on which the following bootstrap assumption holds 
\begin{equation} \mathcal{E}_N(t) \leq 2 C_N \varepsilon (1+t)^\delta \,, \quad \forall \, 0 \leq t \leq T < T_0 \,, \label{eq:thm11.1-bootstrap} \end{equation}
for some fixed $\delta \in (0, 1/4)$ with $\delta < \gamma<1/2$.
Condition \eqref{eq:thm11.1-bootstrap} then allows us to use Propositions \ref{prop-10.1a} and \ref{prop-10.1b} to derive \eqref{eq:thm-11.1-assumpt1} and Proposition \ref{prop-10.2} to derive \eqref{eq:thm-11.1-assumpt2}. 
Theorem \ref{theorem-11.1} then implies that we may choose a sufficiently small $\varepsilon_0$ such that 
$$ \mathcal{E}_N (t) \leq C_N \varepsilon^2 (1+t)^{C \varepsilon} \leq C_N \varepsilon (1+t)^\delta \quad \forall \, t \in [0, T] \,, \quad \forall \, \varepsilon \leq \varepsilon_0 \,.$$
Thus we have contracted the maximality of $T$ and so $(m_\mn + h^0_\mn + h^1_\mn(t), \psi_k(t))$ must be a global-in-time solution.
\end{proof}
\appendix
\section{Non-Minimally Coupled Einstein-Maxwell-Scalar system from the $n=0$ mode reduction} \label{section:ems-full}
In this section we derive the non-minimally coupled Einstein-Maxwell-Scalar system which arises from higher-dimensional vacuum gravity truncated to the zero mode over a $\tee$. The $(3+d+1)-$dimensional coordinates are labelled by Greek indices $\mu, \nu$ and the $(3+1)-$dimensional flat spacetime labelled by Roman letters $a,b$. The compact spacetime is labelled by $\abar, \bbar$, though note that $\alpha, \beta$ are also frequently used in the literature. 
By considering only the $n=0$ modes, the higher-dimensional metric $G_\mn$ depends only on the non-compact coordinates
\begin{equation}
G_\mn(x^a, x^\abar) = G_\mn(x^a) \,. \label{eq:ems-G}
\end{equation}
It is standard to use an Ansatz to parametrise the space of solutions satisfying \eqref{eq:ems-G} in terms of the $(3+1)-$dimensional metric $g_{ab}$, collection of vector potentials $\A_a^\abar$ and dilatons $\hat{g}_\abbar$
\begin{equation}
G_\mn := \begin{pmatrix}
g_{ab} \vert \hat{g} \vert^\lambda + \A_a^\abar \A_b^\bbar \hat{g}_\abbar & \A^\cbar_a \hat{g}_{\bbar \cbar} \\
\A_b^\cbar \hat{g}_{\abar \cbar} & \hat{g}_\abbar
\end{pmatrix} \,.\label{eq:ems-ansatz}
\end{equation}
where $\lambda := - \frac{1}{2}$ and $\vert \hat{g} \vert :=\vert \det \hat{g} \vert$.  
The inverse is now
$$ G^\mn = \begin{pmatrix}
g^{ab} \vert \hat{g} \vert^{- \lambda} & - \A_c^\bbar g^{ac} \vert \hat{g} \vert^{- \lambda} \\
- \A_c^\abar g^{ac} \vert \hat{g} \vert^{- \lambda} & \hat{g}^\abbar + \A_c^\abar \A_d^\bbar g^{cd} \vert \hat{g} \vert^{- \lambda} 
\end{pmatrix} \,. $$
The higher-dimensional wave gauge \eqref{eq:intro-wave-gauge} is equivalent to
$$G^{\nu \rho} \p_\rho G_\mn = \frac{1}{2} G^{\rho \sigma} \p_\mu G_{\rho \sigma} \,. $$
 Under \eqref{eq:ems-G} and \eqref{eq:ems-ansatz} this becomes
\begin{subequations}
\begin{align}
g^{cb} \p_b g_{ac} &= \frac{1}{2} g^{cd} \p_a g_{cd} \,, \label{eq:ems-gauge1}\\  
g^{ab} \p_b \A^\cbar_a &= 0 \,. \label{eq:ems-gauge2}
 \end{align}
\end{subequations}
\eqref{eq:ems-gauge1} is the standard wave-coordinate condition on the $(3+1)-$dimensional non-compact spacetime and \eqref{eq:ems-gauge2} reduces to the standard Lorenz gauge $\nabla^a A_a^\cbar = 0$ using an equivalent form of \eqref{eq:ems-gauge1}, namely $\Gamma^d := g^{ab} \Gamma^d_{ab}=0$. 
To determine the equations of motion, we express the Lagrangian for the Einstein equations in terms of the Ansatz \eqref{eq:ems-ansatz}.
\begin{align*}
\mathcal{L}:=\sqrt{\vert \det G \vert} R[G] &= \sqrt{\vert \det G \vert} \left( \frac{1}{4} G^\mn ( G^{\rho \sigma} \p_\mu G_{\rho \sigma} ) (G^{\lambda \tau} \p_\nu G_{\lambda \tau}) + \frac{1}{2} (\p_\mu G^\mn) G^{\rho \sigma} \p_\nu G_{\rho \sigma} \right. \\
& \qquad  \qquad \qquad \left. + \frac{1}{4} G^\mn \p_\mu G^{\rho \sigma} \p_\nu G_{\rho \sigma} - \frac{1}{2} G^\mn \p_\mu G^{\rho \sigma} \p_\rho G_{\nu \sigma} \right) \\
&=  \sqrt{\vert \det g \vert} \left( R[g] - \frac{1}{8} g^{ab} (\hat{g}^{\cbar \dbar} \p_a g_{\cbar \dbar} ) ( \hat{g}^{\underline{e} \, \underline{f}} \p_b \hat{g}_{\underline{e} \, \underline{f}}) + \frac{1}{4} g^{ab} \p_a \hat{g}^\abbar \p_b \hat{g}_\abbar \right. \\
& \qquad \qquad \qquad \left. - \frac{1}{4} \vert \hat{g} \vert^{1/2} g^{ab} g^{cd} \hat{g}_\abbar \mathcal{F}^\abar _{ac} \mathcal{F}^\bbar_{bd} \right) \,,
\end{align*}
where $\mathcal{F}^\abar_{ab} := \p_a \A^\abar_b - \p_b  \A^\abar_a$. 
Recall the Einstein-Hilbert action is
$$ S_{EH} := \int \mathcal{L}d^\mu x = \int \sqrt{\vert \det G \vert} R[G] d^\mu x\,.$$
Varying $S_{EH}$ with respect to $g_{ab} \,, \A^\abar_a$ and $\hat{g}_\abbar$ respectively we obtain the following equations of motion: 
\begin{subequations}
\begin{align}
R_{ab}[g] - \frac{1}{2} R[g] g_{ab} &= T_{ab} \,,\notag \\
T_{ab} &:= - \frac{1}{8} (\hat{g}^{\cbar \dbar} \p_a g_{\cbar \dbar} ) ( \hat{g}^{\underline{e} \, \underline{f}} \p_b \hat{g}_{\underline{e} \, \underline{f}}) + \frac{1}{16} g_{ab} \left( g^{cd} (\hat{g}^{\cbar \dbar} \p_c g_{\cbar \dbar} ) ( \hat{g}^{\underline{e} \, \underline{f}} \p_d \hat{g}_{\underline{e} \, \underline{f}})\right) \label{eq:ems-eq-g}\\
& \quad + \frac{1}{4} \p_a \hat{g}^\abbar \p_b \hat{g}_\abbar - \frac{1}{8} g_{ab} g^{cd} \p_c \hat{g}^\abbar \p_d \hat{g}_\abbar - \frac{1}{2} \vert \hat{g} \vert^{1/2} g^{cd} \hat{g}_\abbar \mathcal{F}^\abar_{ac} \mathcal{F}^\bbar_{bd} \notag \\
& \quad + \frac{1}{8} \vert \hat{g} \vert^{1/2} g_{ab} g^{cd} g^{ef} \hat{g}_\abbar \mathcal{F}^\abar_{ce} \mathcal{F}^\bbar_{df} \,, \notag \\
\nabla_a \left( \vert \hat{g} \vert^{1/2} \mathcal{F}^{ab}_\abar \right) &= 0 \,, \\
\tbox_g \hat{g}_\abbar &= \frac{1}{2} g^{ab} \p_a \hat{g}_\abbar ( g^{\cbar \, \dbar} \p_b \hat{g}_{\cbar \dbar} )- g^{ab} \hat{g}^{\cbar \, \dbar} \p_a \hat{g}_{\abar \,\cbar} \p_b \hat{g}_{\bbar \,\dbar}- 2 g^{ab} \p_b \hat{g}^{\cbar \, \dbar} \hat{g}_{\dbar \, \bbar} \p_a \hat{g}_{\abar \, \cbar} \\
& \quad - \frac{1}{4} \vert \hat{g} \vert^{1/2} \hat{g}_\abbar \mathcal{F}^\cbar_{ab} \mathcal{F}^{ab}_\cbar + \frac{1}{2} \vert \hat{g} \vert^{1/2} \mathcal{F}_{ab \, \abar} F^{ab}_\bbar - \frac{1}{2} g^{ab} \p_a \left( \hat{g}_\abbar (g^{\cbar \, \dbar} \p_b \hat{g}_{\cbar \, \dbar}) \right) \notag \,.
\end{align} \label{eq:ems-eom}
\end{subequations}
These equations will imply constraint equations for the initial data $(\Sigma_0, \gamma_{ij}, K_{ij}, \mathcal{G}^\abar_{ij}, \mathcal{H}^\abar_{ij}, \hat{\gamma}_\abbar, \hat{K}_\abbar)$ that must be satisfied on the initial surface $\Sigma_0$. Note here $\Sigma_0$ is a 3-dimensional manifold and so the solution we seek $(g_{ab}, \A_a^\abar, \hat{g}_\abbar )$ will be such that $\gamma_{ij}$ is the pull-back of $g_{ab}$ to $\Sigma_0$, $K_{ij}$ is the second fundamental form of $\Sigma_0$ and the restriction of $(\A_a^\abar, \mathcal{F}_{ab}^\abar)$ and $(\hat{g}_\abbar, \p_t \hat{g}_\abbar)$ to $\Sigma_0$ is $(\mathcal{G}^\abar_{ij}, \mathcal{H}^\abar_{ij})$ and $(\hat{\gamma}_\abbar, \hat{K}_\abbar)$ respectively. 

If $n$ is the normal to $\Sigma_0$ normalised such that $g(n, n) = -1$, then the constraint equations are
\begin{equation} \begin{split}
R[\gamma] - K^{ij} K_{ij} + K^i{}_i K^j{}_j &= n^a n^b T_{ab} \vert_{\Sigma_0} \,, \\
D_j K^j {}_i - D_i K^j{}_j &= n^a(g^{bc} - n^b n^c) T_{ac} \vert_{\Sigma_0} \,, \\
D_i \left( \vert \det \hat{\gamma} \vert^{1/2} \mathcal{F}^{i0}_\abar \right) \vert_{\Sigma_0} &= 0 \,,
\end{split} \label{eq:ems-constrainteq} \end{equation} 
where $D$ is the covariant derivative associated with $\gamma$. 
Thus from the lower-dimensional perspective the constraint equations \eqref{eq:ems-constrainteq} can be interpreted as having some initial energy density $\rho :=n^a n^b T_{ab}$ coming from the scalar fields and Maxwell field strengths. 

Theorem \eqref{intro:corol-motiv} implies that $h_\abbar (t) \to 0$ as $t \to \infty$. Thus $\hat{g}_\abbar = G_\abbar (t) \to \delta_\abbar$ as $t \to \infty$ and so the radii of the $\tee$ stabilise to their initial (unperturbed) value. Similarly by inverting the Ansatz \eqref{eq:ems-ansatz} one obtains that the vector potential $\A_a^\abar \to 0$ as $t\to\infty$.

\subsection*{Acknowledgements}
The author greatly thanks Pieter Blue for helpful discussions and advice during this project. 

\section{Hardy Inequality and Further Identities} \label{section:hardy}
In this section we state some important results proved in \cite{LR:04}, omitting the Lemmas which lead to these. 
\begin{corollary}[Corollary 13.3 from \cite{LR:04}] \label{corol-13.3}
If $\gamma >0$ and $\mu >0$, then for any $a \in [-1, 1]$ and any $\phi \in C^\infty_0(\R^3)$ we have
$$ \int \frac{\vert \phi \vert^2}{(1+\qv)^2} \frac{w dx}{(1+t+\qv )^{1-a}} \lesssim \int \vert \p \phi \vert^2 \frac{w dx}{(1+t+\qv )^{1-a}} \,. $$
If additionally $a<2 \min (\gamma, \mu)$ then
$$ \int \frac{\vert \phi \vert^2}{(1+\qv)^2} \frac{(1+\qv)^{-a}}{(1+t+\qv)^{1-a}} \frac{w dx}{(1+q_- )^{2 \mu}} \lesssim \int \vert \p \phi \vert^2 \min \left( w', \frac{w}{(1+t+\qv )^{1-a}} \right) dx \,,$$
where $q_- = \qv$ when $q<0$ and $q_- = 0$ when $q>0$. 
\end{corollary}

We next briefly state some key results from \cite{LR:04}, the first two making use of the null frame. Recall the definition of the reduced wave operator $\tbox_g := g^\mn \p_\mu \p_\nu$ where $g_\mn$ is the full (perturbed) metric. The final Proposition \ref{prop-5.3} gives control on the commutator  $\tbox_g Z^I \phi - \hat{Z}^I \tbox_g \phi$.
\begin{lemma}[Lemma 4.2 from \cite{LR:04}] \label{lemma-4.2a}
For an arbitrary 2-tensor $\pi$ and scalar function $\phi$, we have
\begin{align*}
\vert \pi^\mn \p_\mu \phi \p_\nu \phi \vert &\lesssim \vert \pi \vert_{\mathcal{L} \mathcal{L}} \vert \p \phi \vert^2 + \vert \pi \vert_\UU \vert \p \phi \vert \vert \overline{\p} \phi \vert \\
\vert L_\mu \pi^\mn \p_\nu \phi \vert & \lesssim \vert \pi \vert_{\mathcal{L} \mathcal{L}} \vert \p \phi \vert + \vert \pi \vert_\UU \vert \overline{\p} \phi \vert \\
\vert ( \p_\mu \pi^\mn ) \p_\nu \phi \vert & \lesssim \left( \vert \p \pi \vert_{\mathcal{L} \mathcal{L}} +  \vert \overline{\p} \pi \vert_\UU   \right) \vert \p \phi \vert +\vert \p \pi  \vert_\UU  \vert \overline{\p} \phi \vert \\
\vert \pi^\mn \p_\mu \p_\nu \phi \vert & \lesssim \vert \pi \vert_{\mathcal{L} \mathcal{L}} \vert \p^2 \phi \vert + \vert \pi \vert_\UU \vert \overline{\p} \p \phi \vert 
\end{align*}
\end{lemma}

\begin{lemma}[Lemma 5.1 from \cite{LR:04}] \label{lemma-5.1}
If $\phi$ is a scalar and $\pi$ a symmetric 2-tensor then
\begin{align}
( 1+t+\vert q \vert ) \vert \pgood \phi \vert + (1+\vert q \vert) \vert \p \phi \vert & \leq C \sum_{\vert I \vert = 1} \vert Z^I \phi \vert \\
\vert  \pgood^2 \phi \vert + r^{-1} \vert \pgood \phi \vert & \leq \frac{C}{r} \sum_{\vert I \vert \leq 2} \frac{\vert Z^I \phi \vert}{1+t + \vert q \vert} \\
\vert \pi^{\mu \nu} \p_\mu \p_\nu \phi \vert &\leq C \left( \frac{\vert \pi \vert_\UU}{1+t+\vert q \vert} + \frac{\vert \pi \vert_\LL}{1+\vert q \vert} \right) \sum_{\vert I \vert \leq 1} \vert \p Z^I \phi \vert \label{eq:lemma5.1-3rd-eq}
\end{align}
\end{lemma}
We introduce the shifted vector field $\hat{Z} := Z + c_Z$ where $c_Z$ is a constant defined by $\p_\mu Z_\nu + \p_\nu Z_\mu = c_Z m_\mn$. In particular, $c_Z = 0$ except when $Z = S$ in which case $c_S = 2$. Note this definition implies the commutation property $[Z, \tbox_m] = -c_Z \tbox_m$ and thus $\tbox_m Z \phi = \hat{Z} \tbox_m \phi$. 

\begin{proposition}[Proposition 5.3 from \cite{LR:04}] \label{prop-5.3}
Recall  $\tbox_g := g^\mn \p_\mu \p_\nu = \tbox_m + H^\mn \p_\mu \p_\nu$. Then for any $Z \in \mathcal{Z}$ it holds
\begin{align*}
\vert \tbox_g Z^I \phi - \hat{Z}^I \tbox_g \phi \vert \lesssim \Big( \frac{\vert ZH \vert_\UU + \vert H \vert_\UU}{1+t+\qv} + \frac{\vert ZH \vert_\LL + \vert H \vert_\LT}{1+\qv} \Big) \sum_{\vert I \vert \leq 1 } \vert \p Z^I \phi \vert \,.
\end{align*}
Furthermore
\begin{align*}
\vert \tbox_g Z \phi - \hat{Z} \tbox_g \phi \vert & \lesssim \frac{1}{1+t+\qv} \sum_{ \vK  \leq  \I } \sum_{\J + (\vK-1)_+ \leq \I} \vert Z^J H \vert_\UU \vert \p Z^K \phi \vert  \\ & + \frac{1}{1+\qv}  \sum_{\vK \leq \I}\Big( \sum_{\J + (\vK-1)_+ \leq \I} \vert Z^J H \vert_\LL + \sum_{\J  + (\vK-1)_+ \leq \I-1} \vert Z^J H \vert_\LT \\ & \qquad \qquad + \sum_{\J + (\vK-1)_+ \leq \I-2} \vert Z^J H \vert_\UU \Big) \vert \p Z^K \phi \vert \,,
\end{align*}
where $(\vK - 1)_+ = \vK - 1$ if $\vK \geq 1$ and $(\vK-1)_+ = 0$ if $\vK = 0$. 
\end{proposition}

\bibliographystyle{plain}
\bibliography{bibliography}
\end{document}